\RequirePackage{fix-cm}
\documentclass[twocolumn]{svjour3}          
\smartqed  

\usepackage[utf8]{inputenc}
\DeclareUnicodeCharacter{00A0}{ } 

\usepackage{graphicx}
\usepackage{amsmath}
\usepackage{amssymb}
\usepackage{commath}
\usepackage{mathtools}
\usepackage{bbm}
\usepackage{verbatim}
\usepackage{color}
\usepackage[dvipsnames]{xcolor}
\usepackage{microtype}
\usepackage[hyphens]{url}

\usepackage{makecell,colortbl}

\usepackage{tikz}
\usepackage{pgfplots}

\usepackage{subcaption}
\pgfplotsset{compat=newest}
\pgfplotsset{plot coordinates/math parser=false}

\usetikzlibrary{decorations.pathreplacing}
\usetikzlibrary{patterns}
\usepgfplotslibrary{groupplots}

\usetikzlibrary{external}
\newlength{\figwidth}
\newlength{\figheight}
\definecolor{gray1}{gray}{0.0}
\definecolor{gray2}{gray}{0.25}
\definecolor{gray3}{gray}{0.5}
\definecolor{gray4}{gray}{0.7}
\definecolor{gray5}{gray}{0.9}

\newlength{\prel}

\pgfplotsset{
  title style = {font=\small},
}

\newcommand{\T}{\textsf{T}} 
\newcommand{\perm}{\text{Perm}} 
\newcommand{\permsc}{\mathrm{Perm}^{\text{\tiny{SC}}}} 

\renewcommand{\b}[1]{\pmb{#1}} 

\DeclareMathOperator*{\argmin}{arg\,min}  

\newcommand{\rev}[1]{#1} 

\begin{document}

\title{Symmetry Exploits for Bayesian Cubature Methods}

\author{Toni Karvonen \and Simo S\"{a}rkk\"{a} \and Chris. J. Oates}

\institute{T. Karvonen \at
              Department of Electrical Engineering and Automation \\
              Aalto University, Finland \\
              \email{toni.karvonen@aalto.fi}
           \and
           S. S\"{a}rkk\"{a} \at
              Department of Electrical Engineering and Automation \\
           Aalto University, Finland \\
          \email{simo.sarkka@aalto.fi}
		  \and
                  C. J. Oates \at
                  School of Mathematics, Statistics and Physics \\
          Newcastle University, UK \\
          \email{chris.oates@ncl.ac.uk}
}

\date{Received: date / Accepted: date}

\maketitle

\begin{abstract}
Bayesian cubature provides a flexible framework for numerical integration, in which {\it a priori} knowledge on the integrand can be encoded and exploited.
This additional flexibility, compared to many classical cubature methods, comes at a computational cost which is cubic in the number of evaluations of the integrand.
It has been recently observed that fully symmetric point sets can be exploited in order to reduce -- in some cases substantially -- the computational cost of the standard Bayesian cubature method.
This work identifies several additional symmetry exploits within the Bayesian cubature framework.
In particular, we go beyond earlier work in considering non-symmetric measures and, in addition to the standard Bayesian cubature method, present exploits for the Bayes--Sard cubature method and the multi-output Bayesian cubature method.

\keywords{probabilistic numerics \and numerical integration \and Gaussian processes \and fully symmetric sets} 
\end{abstract}

\section{Introduction}
 
This paper considers the numerical approximation of an integral
\begin{eqnarray*}
I(f^\dagger) \coloneqq \int_M f^\dagger(\b{x}) \mathrm{d}\nu(\b{x}),
\end{eqnarray*}
where $(M, \mathcal{B}, \nu)$ is a Borel probability space with $M$ any Borel measurable non-empty subset of $\mathbb{R}^m$ and $f^\dagger \colon M \rightarrow \mathbb{R}$ is a $\mathcal{B}$-measurable scalar-valued integrand (vector-valued integrands will be considered in Section~\ref{subsec: vvbc}). Additional assumptions will be made when necessary.
Our interest is in the situation where the exact values of $f^\dagger$ cannot be deduced until the function itself is evaluated, and that the evaluations are associated with a substantial computational cost or a very large number of them is required.
Such situations are typical in, for example, uncertainty quantification for chemical systems \cite{Najm2009}, fluid mechanical simulation \cite{Xiu2003} and certain financial applications \cite{Holtz2011}.

In \rev{the presence} of a limited computational budget, it is natural to exploit any contextual information that may be available on the integrand.
Classical cubatures, such as spline-based or Gaussian cubatures, are able to exploit abstract mathematical information, such as the number of continuous derivatives of the integrand \cite{Davis2007}.
However, in situations where more detailed or specific contextual information is available to the analyst, the use of generic classical cubatures can be sub-optimal.
 
The language of probabilities provides one mechanism in which contextual information about the integrand can be captured.
Let $(\Omega,\mathcal{F},\mathbb{P})$ be a probability space.
Then an analyst can elicit their prior information about the integrand $f^\dagger$ in the form of a stochastic process model
\begin{eqnarray}
\omega \mapsto f(\cdot \, ; \omega), \qquad \omega \in \Omega,
\label{eq: prior stochastic process}
\end{eqnarray}
wherein the function $\b{x} \mapsto f(\b{x} ; \omega)$ is $\mathcal{B}$-measurable for each fixed $\omega \in \Omega$.
Through the stochastic process, the analyst can encode both abstract mathematical information, such as the number of continuous derivatives of the integrand, and specific contextual information, such as the possibility of a trend or a periodic component.
The process of elicitation is not discussed in this work (see~\cite{Diaconis1988,Hennig2015}); for our purposes the stochastic process in~\eqref{eq: prior stochastic process} is considered to be provided.
 
In Bayesian cubature methods, due to Larkin~\cite{Larkin1972} and re-discovered in \cite{Diaconis1988,OHagan1991,Minka2000}, the analyst first selects a point set $X = \{\b{x}_i\}_{i=1}^N \subset M$, $N \in \mathbb{N}$, on which the true integrand $f^\dagger$ is evaluated.
Let this data be denoted $\mathcal{D} = \{(\b{x}_i,f^\dagger(\b{x}_i))\}_{i=1}^N$.
Then the analyst conditions their stochastic process according to these data $\mathcal{D}$, to obtain a second stochastic process
\begin{eqnarray*}
\omega \mapsto f_N(\cdot \, ; \omega) .
\end{eqnarray*}
The analyst reports the implied distribution over the value of the integral of interest; that is the law of the random variable
\begin{eqnarray*}
\omega \mapsto \int_M f_N(\b{x} ; \omega) \mathrm{d}\nu(\b{x}) .
\end{eqnarray*}
This distribution can be computed in closed form under certain assumptions on the structure of the prior model.
A sufficient condition is that the stochastic process is Gaussian, which (arguably) does not severely restrict the analyst in terms of what contextual information can be included \cite{Rasmussen2006}.
In addition, the probabilistic output of the method enables uncertainty quantification for the unknown true value of the integral~\mbox{\cite{Larkin1972,Cockayne2017,Briol2017}}.
These appealing properties have led to Bayesian cubature methods being used in diverse areas such as from computer graphics \cite{Marques2013}, non-linear filtering \cite{Prueher2017} and applied Bayesian statistics~\cite{Osborne2012}.
 
The theoretical aspects of Bayesian cubature methods have now been widely-studied.
In particular, convergence of the posterior mean point estimator
\begin{eqnarray}
\int_\Omega \int_M f_N(\b{x} ; \omega) \mathrm{d}\nu(\b{x}) \mathrm{d}\mathbb{P}(\omega) \rightarrow \int_M f^\dagger(\b{x}) \mathrm{d}\nu(\b{x})
\label{eq: posterior mean}
\end{eqnarray}
as $N \rightarrow \infty$ has been studied in both the well-specified \cite{Bezhaev1991,SommarivaVianello2006,Briol2015,Ehler2018,Briol2017} and mis-specified \cite{Kanagawa2016,Kanagawa2017} regimes.
Some relationships between the posterior mean estimator and classical cubature methods have been documented in \cite{Diaconis1988,Sarkka2016,Karvonen2017a}.
In \cite{Larkin1974,OHagan1991,Karvonen2018c} the {\it Bayes--Sard} framework was studied, where it was proposed to incorporate an explicit parametric component~\cite{OHagan1978} into the prior model in order that contextual information, such as trends, can be properly encoded.
The choice of point set $X$ for Bayesian cubature has been studied in \cite{Briol2015,Bach2017,Briol2017a,Oettershagen2017,Chen2018,Pronzato2018}.
In addition, several extensions have been considered to address specific technical challenges posed by non-negative integrands~\cite{Chai2018}, model evidence integrals in a Bayesian context~\cite{Osborne2012,Gunter2014}, ratios~\cite{Osborne2012a}, non-Gaussian prior models~\cite{Kennedy1998,Prueher2017}, measures that can be only be sampled~\cite{Oates2017}, and vector-valued integrands~\cite{Xi2018}.

Despite these recent successes, a significant drawback of Bayesian cubature methods is that the cost of computing the distributional output is typically cubic in $N$, the size of the point set.
For integrals whose domain $M$ is high-dimensional, the number $N$ of points required can be exponential in $m = \text{dim}(M)$.
Thus the cubic cost associated with Bayesian cubature methods can render them impractical.
In recent work, Karvonen and Särkkä~\cite{Karvonen2018} noted that symmetric structure in the point set can be exploited to reduce the total computational cost.
Indeed, in some cases the exponential dependence on $m$ can be reduced to (approximately) linear.
This is a similar effect to that achieved in the circulant embedding approach \cite{Dietrich1997}, or by the use of $\mathcal{H}$-matrices \cite{Hackbusch1999} and related approximations \cite{Schaefer2017}, though the approaches differ at a fundamental level.
The aim of this paper is to present several related \emph{symmetry exploits} that are specifically designed to reduce computational cost of Bayesian cubature methods.

Our principal contributions are following:
\emph{First}, the techniques developed in \cite{Karvonen2018} are extended to the Bayes--Sard cubature method.
This results in a computational method that is, essentially, of the complexity \sloppy{${\mathcal{O}(J^3 + JN)}$}, where $J$ is the number of symmetric sets that constitute the full point set, instead of being cubic in $N$.
In typical scenarios there are at most a few hundred symmetric sets even though the total number of points can go up to millions.
\emph{Second}, we present an extension to the multi-output (i.e., vector-valued) Bayesian cubature method that is used to simultaneously integrate $D \in \mathbb{N}$ related integrals.
In this case, the computational complexity is reduced from $\mathcal{O}(D^3 N^3)$ to $\mathcal{O}(D^3 J^3 + DJN)$.
\emph{Third}, a symmetric change of measure technique is proposed to avoid the (strong) assumption of symmetry on the measure $\nu$ that was required in~\cite{Karvonen2018}.
\emph{Fourth}, the performance of our techniques is empirically explored.
Throughout, our focus is not on the performance of these integration methods, which has been explored in earlier work, already cited. 
Rather, our focus is on how computation for these methods can be accelerated.

The remainder of the article is structured as follows:
Section~\ref{sec:background} covers the essential background for Bayesian cubature methods and introduces fully symmetric sets that are used in the symmetry exploits throughout the article.
Sections~\ref{sec:bsc-fss} and~\ref{subsec: vvbc} develop fully symmetric Bayes--Sard cubature and fully symmetric multi-output Bayesian cubature.
Section~\ref{sec:is-trick} explains how the assumption that $\nu$ is symmetric can be relaxed.
In Section~\ref{sec:results} a detailed selection of empirical results are presented.
Finally, some concluding remarks and discussion are contained in Section~\ref{sec:conclusion}.

\section{Background} \label{sec:background}

This section reviews the standard Bayesian cubature method, due to Larkin~\cite{Larkin1972}, and explains how fully symmetric sets can be used to alleviate its computational cost, as proposed in \cite{Karvonen2018}.

\subsection{Standard Bayesian Cubature} \label{subsec: SBC}

In this section we present explicit formulae for the Bayesian cubature method in the case where the prior model~\eqref{eq: prior stochastic process} is a Gaussian random field.
To simplify the notation, Sections~\ref{sec:background} and~\ref{sec:bsc-fss} assume that the integrand has scalar output (i.e. $D = 1$); this is then extended to vector-valued output in Section~\ref{subsec: vvbc}.

To reduce the notational overhead, in what follows the $\omega \in \Omega$ argument is left implicit.
Thus we consider $f(\b{x})$ to be a scalar-valued random variable for each $\b{x} \in M$.
In particular, in this paper we focus on stochastic processes that are Gaussian, meaning that there exists a \emph{mean function} $m : M \rightarrow \mathbb{R}$ and a symmetric positive definite \emph{covariance function} (or \emph{kernel}) $k : M \times M \rightarrow \mathbb{R}$ such that $[f(\b{x}_1), \ldots, f(\b{x}_N)]^\T \in \mathbb{R}^N$ has the multivariate Gaussian distribution
\begin{equation*}
\mathrm{N}\left( \begin{bmatrix} m(\b{x}_1) \\ \vdots \\ m(\b{x}_N) \end{bmatrix}, \begin{bmatrix} k(\b{x}_1, \b{x}_1) & \cdots & k(\b{x}_1,\b{x}_N) \\ \vdots & \ddots & \vdots \\ k(\b{x}_N,\b{x}_1) & \cdots & k(\b{x}_N,\b{x}_N) \end{bmatrix} \right)
\end{equation*}
for any $N \in \mathbb{N}$ and all point sets $\{\b{x}_i\}_{i=1}^N \subset M$. We assume that $\int_M k(\b{x}, \b{x}) \dif \nu(\b{x}) < \infty$. 

The conditional distribution $f_N$ of this field, based on the data $\mathcal{D} = \{(\b{x}_i,f^\dagger(\b{x}_i)\}_{i=1}^N$ of function evaluations at the points $X = \{\b{x}_i\}_{i=1}^N$, is also Gaussian, with mean and covariance functions
\begin{align}
m_N(\b{x}) & = m(\b{x}) + \b{k}_X(\b{x})^\T \b{K}_{X}^{-1} (\b{f}^\dagger_X - \b{m}_X), \label{eqn:GP-mean} \\
k_N(\b{x},\b{x}') & = k(\b{x},\b{x}') - \b{k}_X(\b{x})^\T \b{K}_{X}^{-1} \b{k}_X(\b{x}'), \label{eqn:GP-var}
\end{align}
where the vector $\b{f}^\dagger_X \in \mathbb{R}^N$ contains evaluations of the integrand, $[\b{f}^\dagger_X]_i = f^\dagger(\b{x}_i)$, the vector $\b{m}_X \in \mathbb{R}^N$ contains evaluations of the prior mean, $[\b{m}_X]_i = m(\b{x}_i)$, the vector $\b{k}_X(\b{x}) \in \mathbb{R}^N$ contains evaluations of the kernel, $[\b{k}_X(\b{x})]_i = k(\b{x},\b{x}_i)$, and $\b{K}_X  = \b{K}_{X,X} \in \mathbb{R}^{N \times N}$ is the \emph{kernel matrix}, $[\b{K}_X]_{ij} = k(\b{x}_i,\b{x}_j)$.
From the fact that linear functionals of Gaussian processes are Gaussian, we obtain that
\begin{eqnarray}
\int_M f_N(\b{x}) \dif \nu(\b{x}) \sim \mathrm{N}\big( \mu_N(f^\dagger) , \sigma_N^2 \big) \label{eq: BC output},
\end{eqnarray}
with
\begin{align}
\mu_N(f^\dagger) ={}& I(m) + \b{k}_{\nu,X}^\T \b{K}_{X}^{-1} (\b{f}^\dagger_X - \b{m}_X), \label{eq: BC mean} \\
\sigma_N^2 ={}& k_{\nu,\nu} - \b{k}_{\nu,X}^\T \b{K}_{X}^{-1} \b{k}_{\nu,X}. \label{eq: BC var} 
\end{align}
Here $k_\nu(\b{x}) := \int_M k(\b{x},\b{x}') \dif \nu(\b{x}')$ is called the \emph{kernel mean} function \cite{Smola2007} and $\b{k}_{\nu,X} \in \mathbb{R}^N$ is the column vector with $[\b{k}_{\nu,X}]_i = k_\nu(\b{x}_i)$, whilst $k_{\nu,\nu} := \int_M k_\nu(\b{x}) \dif \nu(\b{x}) \geq 0$ is the variance of the integral itself under the prior model.
The assumption $\int_M k(\b{x}, \b{x}) \dif \nu(\b{x}) < \infty$ guarantees that the kernel mean is finite.
This method is known as the \emph{standard Bayesian cubature}, with the implicit understanding that the model for the integrand should be carefully selected to ensure~\eqref{eq: BC output} is well-calibrated \cite{Briol2017}, meaning that the uncertainty assessment can be trusted.
The need for careful calibration is in line with standard approaches to the Gaussian process regression task \cite{Rasmussen2006}.

To understand when the Bayesian cubature output is meaningful, it is useful to write the posterior mean and variance~\eqref{eq: BC mean} and~\eqref{eq: BC var} in terms of the \emph{weight} vector
\begin{equation}\label{eqn:BC-weights}
\b{w}_X \coloneqq \b{K}_X^{-1} \b{k}_{\nu,X}.
\end{equation}
That is, we have $\mu_N(f^\dagger) = I(m) + \b{w}_X^\T (\b{f}_X^\dagger - \b{m}_X)$ and $\sigma_N^2 = k_{\nu,\nu} - \b{w}_X^\T \b{k}_{\nu,X}$.
Let $\mathcal{H}(k)$ be the Hilbert space reproduced by the kernel $k$ (see \cite{Berlinet2011} for background).
It can then be verified that $\b{w}_X$ solves a quadratic minimisation problem of approximating $k_\nu$ with a function from the finite-dimensional space spanned by \sloppy{${\{k(\cdot, \b{x})\}_{\b{x} \in X} \subset \mathcal{H}(k)}$}, namely:
\begin{equation*}
\b{w}_X = \argmin_{\b{w} \in \mathbb{R}^N} \bigg\| k_\nu(\cdot) - \sum_{i=1}^N w_i k(\cdot, \b{x}_i) \bigg\|_{\mathcal{H}(k)}.
\end{equation*}
and that the minimum the value of this norm is $\sigma_X$ (see e.g. \cite[Ch.\ 3]{Oettershagen2017} and \cite{Bach2012}).
Equivalently, the weight vector can be obtained as the minimiser of the worst case error 
\begin{eqnarray*}
\sup_{\|f^\dagger\|_{\mathcal{H}(k)} \leq 1} \abs[3]{ \int_M f^\dagger(\b{x}) \dif \nu(\b{x}) - \sum_{i=1}^N w_i f^\dagger(\b{x}_i) }
\end{eqnarray*}
among all cubature rules with points $X$, with $\sigma_N$ corresponding to the minimal worst case error \cite{Briol2017,Oettershagen2017}.
Thus, in terms of uncertainty quantification, the posterior standard deviation $\sigma_X$ can indeed be meaningfully related to the integration problem being solved.

The principal motivation for this work is the observation that both~\eqref{eq: BC mean} and~\eqref{eq: BC var} involve the solution of an $N$-dimensional linear system defined by the matrix $\b{K}_{X}$.
In general this is a dense matrix and, as such, in the absence of additional structure in the linear system~\cite{Karvonen2018} or further approximations (e.g. \cite{LazaroGredilla2010,Hensman2018,Schaefer2017}), the computational complexity associated with the standard Bayesian cubature method is $\mathcal{O}(N^3)$.
Moreover, it is often the case that $\b{K}_X$ is ill-conditioned \cite{Schaback1995,Stein2012}.
The exploitation of symmetric structure to circumvent the solution of a large and ill-conditioned linear system would render Bayesian cubature more practical, in the sense of computational efficiency and numerical robustness; this is the contribution of the present article.

\subsection{Symmetry Properties}

\begin{figure}[t]
\centering
  \begin{subfigure}[b]{0.49\columnwidth}
  \centering
  \includegraphics{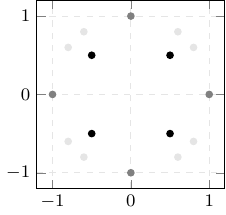}
  \end{subfigure}
  \begin{subfigure}[b]{0.49\columnwidth}
  \centering
  \includegraphics{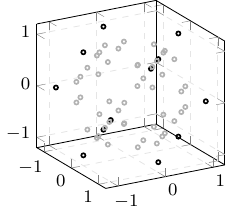}
  \end{subfigure}
  \caption{Fully symmetric sets generated by the vectors $[0.5,0.5]$, $[1,0]$, and $[0.6,0.8]$ in $\mathbb{R}^2$ (left) and $[1,1,0]$ and $[0.2,0.6,0.8]$ in $\mathbb{R}^3$ (right).}\label{fig:fss-example}
\end{figure}

Next we introduce fully symmetric sets and related symmetry concepts, before explaining in Section~\ref{sec:bc-fss} how these can be exploited for computational simplification in the standard Bayesian cubature method.
Note that, in what follows, no symmetry properties are needed for the integrand $f^\dagger$ itself.

\subsubsection{Fully Symmetric Point Sets}

Given a vector $\b{\lambda} \in \mathbb{R}^m$, the \emph{fully symmetric set} $[\b{\lambda}] \subset \mathbb{R}^m$ generated by this vector is defined as the point set consisting of all vectors that can be obtained from $\b{\lambda}$ via coordinate permutations and sign changes.
That is,
\begin{equation*}
\begin{split}
[\b{\lambda}] &= [\lambda_1,\ldots,\lambda_d] \\
&\coloneqq \bigcup_{q \in \Pi_m} \bigcup_{\b{s} \in S_m} \big\{ (s_1 \lambda_{q_1}, \ldots, s_d \lambda_{q_d} \big\} \subset \mathbb{R}^m,
\end{split}
\end{equation*}
where $\Pi_m$ and $S_m$ stand for the collections of all permutations of the first $m$ positive integers and of all vectors of the form $\b{s} = (s_1,\ldots,s_m)$ with each $s_i$ either $1$ or $-1$.
Here $\b{\lambda}$ is called a \emph{generator vector} and its individual elements are called \emph{generators}.
Alternatively, we can write the fully symmetric set in terms of permutation and sign change matrices:
\begin{equation*}
[\b{\lambda}] = \bigcup_{\b{P} \in \permsc_m} \b{P} \b{\lambda},
\end{equation*}
where $\permsc_m$ is the collection of $m \times m$ matrices having exactly one non-zero element on each row and column, this element being either $1$ or $-1$.
Some fully symmetric sets are displayed in Figure~\ref{fig:fss-example}.
The cardinality of a fully symmetric set $[\b{\lambda}]$, generated by a generator vector $\b{\lambda}$ containing $r_0$ zero generators and $l$ distinct non-zero generators with multiplicities $r_1,\ldots,r_l$, is
\begin{equation}\label{eq:fss-size}
\#[\b{\lambda}] = \frac{2^{m-r_0}d!}{r_0! \cdots r_l!}.
\end{equation}
See Table \ref{table:sizes} for a number of examples in low dimensions.

For $\b{\lambda} \in \mathbb{R}^m$ having non-negative elements, we occasionally need the concept of a \emph{non-negative} fully symmetric set
\begin{equation*}
[\b{\lambda}]^+ \coloneqq \bigcup_{\b{P} \in \perm_m} \b{P}\b{\lambda} \subset \mathbb{R}^m_+,
\end{equation*}
where $\perm_m \subset \permsc_m$ is the collection of $m \times m$ permutation matrices.

\begin{table}[t]
\begin{center}
\captionof{table}{Sizes of fully symmetric sets generated by the generator vector $\b{\lambda} = (\lambda_1,\ldots,\lambda_l,0,\ldots,0)$ having $l \leq m$ distinct non-zero elements $\lambda_1,\ldots,\lambda_l$ (see~\eqref{eq:fss-size}).
}\label{table:sizes}
\small
\caption*{\small{\textbf{Dimension} ($m$)}}\vspace{-0.3cm}
\begin{tabular}[t!]{c|c c c c c c c c c c}
\arrayrulecolor{lightgray}
 & 2 & 3 & 4 & 5 & 6 & 7 \\
\Xhline{1pt}
$l=1$ & 4 & 6 & 8 & 10 & 12 & 14\\ \hline
$l=2$ & 8 & 24 & 48 & 80 & 120 & 168\\ \hline
$l=3$ & - & 48 & 192 & 480 & 960 & 1,680\\ \hline
$l=4$ & - & - & 384 & 1,920 & 5,760 & 13,440\\ \hline
$l=5$ & - & - & - & 3,840 & 23,040 & 80,640\\ \hline
$l=6$ & - & - & - & - & 46,080 & 322,560\\ \hline
$l=7$ & - & - & - & - & - & 645,120\\
\end{tabular}
\end{center}
\end{table}

\subsubsection{Fully Symmetric Domains, Kernels, and Measures} \label{sec:fss-objects}

At this point we introduce several related definitions; these enable us later to state precisely which symmetry assumptions are being exploited.

\paragraph{Domains.} It will be assumed in the sequel that \sloppy{${M \subset \mathbb{R}^m}$} is a \emph{fully symmetric domain}, meaning that every fully symmetric set generated by a vector from $M$ is contained in $M$: $[\b{\lambda}] \subset M$ whenever $\b{\lambda} \in M$.
Equivalently, \sloppy{${M = \b{P}M = \{ \b{P} \b{x} \, \colon, \b{x} \in M \}}$} for any $\b{P} \in \permsc_m$.
Most popular domains, such as the whole of $\mathbb{R}^m$, hypercubes of the form $[-a, a]^m$ (from which e.g. the unit hypercube can be obtained by simple translation and scaling), balls and spheres, are fully symmetric.

\paragraph{Kernels.} A kernel $k \colon M \times M \to \mathbb{R}$ defined on a fully symmetric domain $M$ is said to be a \emph{fully symmetric kernel} if $k(\b{P}\b{x},\b{P}\b{x}') = k(\b{x},\b{x}')$ for any $\b{P} \in \permsc_m$.
Basic examples of fully symmetric kernels include isotropic kernels and products and sums of isotropic one-dimensional kernels.

\paragraph{Measures.} A measure $\nu$ on a fully symmetric domain $M$ is a \emph{fully symmetric measure} \rev{if it is invariant under fully symmetric pushforwards: $\b{P}_*(\nu) = \nu$ for any \sloppy{${\b{P} \in \permsc_m}$}.}
If $\nu$ admits a Lebesgue density $p_\nu$, this condition is equivalent to $p_\nu(\b{x}) = p_\nu(\b{P}\b{x})$ for any $\b{P} \in \permsc_m$.
Note that this is a narrow class of measures and a relaxation of this assumption is discussed in Section~\ref{sec:is-trick}.

\subsubsection{Fully Symmetric Cubature Rules}

The linear functional $\mu(f^\dagger) = \sum_{i=1}^N w_i f^\dagger(\b{x}_i)$ is said to be \emph{fully symmetric cubature rule} if its point set can be written as a union of a number $J \in \mathbb{N}$ of fully symmetric sets $[\b{\lambda}^1], \ldots, [\b{\lambda}^J]$ and all points in each $[\b{\lambda}^j]$ are assigned an equal weight.
That is, a fully symmetric cubature rule is of the form
\begin{equation*}
\mu(f^\dagger) = \sum_{j=1}^J w^{\text{\tiny{FS}}}_j \sum_{\b{x} \in [\b{\lambda}^j]} f^\dagger(\b{x})
\end{equation*}
for some weights $\b{w}^{\text{\tiny{FS}}} \in \mathbb{R}^J$ and generator vectors \sloppy{${\b{\lambda}^1 , \ldots , \b{\lambda}^J \in M}$}. 
Because this structure typically greatly simplifies design of the weights, many classical polynomial-based cubature rules are fully symmetric~\cite{McNameeStenger1967,Genz1986,GenzKeister1996,LuDarmofal2004}, including certain sparse grids~\cite{NovakRitter1999,NovakRitter1999b}

\subsection{Fully Symmetric Bayesian Cubature} \label{sec:bc-fss}

The central aim of this article is to derive generalisations for the Bayes--Sard and multi-output Bayesian cubatures of the following result from~\cite{Karvonen2018}, originally developed only for the standard Bayesian cubature method. 

\begin{theorem} \label{thm:bc-fss}
Consider the standard Bayesian cubature method based on a domain $M$, measure $\nu$, and kernel $k$ that are each fully symmetric and fix the mean function to be $m \equiv 0$.
Suppose that the point set is a union of $J$ fully symmetric sets: $X = \bigcup_{j=1}^J [\b{\lambda}^j]$ for some distinct generator vectors \sloppy{${\Lambda = \{ \b{\lambda}^1,\ldots,\b{\lambda}^J \} \subset M}$}.
Then the output of the standard Bayesian cubature method can be expressed in the fully symmetric form
\begin{align*}
\mu_N(f^\dagger) &= \sum_{j=1}^J w_{\Lambda,j} \sum_{\b{x} \in [\b{\lambda}^j]} f^\dagger(\b{x}), \\
\sigma_N^2 &= k_{\nu,\nu} - \sum_{j=1}^J w_{\Lambda,j} k_\nu(\b{\lambda}^j) \#[\b{\lambda}^j].
\end{align*}
The weights $\b{w}_\Lambda \in \mathbb{R}^J$ are the solution to the linear system $\b{S} \b{w}_\Lambda = \b{k}_{\nu,\Lambda}$ of $J$ equations, where
\begin{equation*}
[\b{S}]_{ij} = \sum_{\b{x} \in [\b{\lambda}^j]} k(\b{\lambda}^i,\b{x}) \quad \text{and} \quad [\b{k}_{\nu,\Lambda}]_j = k_\nu(\b{\lambda}^j).
\end{equation*}
\end{theorem}

Theorem \ref{thm:bc-fss} demonstrates the principal idea; that one can exploit symmetry to reduce the number of kernel evaluations needed in the standard Bayesian cubature method from $N^2$ to $NJ$ and decrease the number of equations in the linear system that needs to be solved from $N$ to $J$.
Since $J$ is typically considerably smaller than $N = \sum_{j=1}^J \#[\b{\lambda}^j]$, using fully symmetric sets results in a substantial reduction in computational cost.
Numerical examples in~\cite{Karvonen2018} showed that sets containing up to tens of millions of points become feasible in the standard Bayesian cubature method when symmetry exploits are used.
The aim of this paper is to generalise these techniques to the important cases of Bayes--Sard cubature (Section~\ref{sec:bsc-fss}) and multi-output Bayesian cubature (Section~\ref{subsec: vvbc}).

\begin{remark} 
If $\#[\b{\lambda}^1] = \cdots = \#[\b{\lambda}^J]$, the condition number of the matrix $\b{S}$ in Theorem~\ref{thm:bc-fss} cannot exceed that of $\b{K}_X$ (similar results are available for the matrices in Theorems~\ref{thm:BSC-FSS} and~\ref{thm:mobc-fss}). 
This scenario occurs in, for instance, the numerical example of Section~\ref{sec:global-illu-int}.
To verify the claim, observe that by Lemma~\ref{lemma: fss-cross-sets} $\b{S} \b{v} = \alpha \b{v}$ implies that the block vector
\begin{equation*}
\b{v}' = \begin{bmatrix} v_1 \b{1}_{\#[\b{\lambda}^1]} \\ \vdots \\ v_J \b{1}_{\#[\b{\lambda}^J]} \end{bmatrix}
\end{equation*}
satisfies $\b{K}_X \b{v}' = \alpha \b{v}'$. 
Consequently, the spectrum of $\b{S}$ is a subset of that of $\b{K}_X$. 
Furthermore, when \sloppy{${\#[\b{\lambda}^1] = \cdots = \#[\b{\lambda}^J]}$}, the matrix $\b{S}$ is symmetric; therefore its condition number is the ratio of the largest and smallest eigenvalues. 
It follows that the condition number of $\b{S}$ must be smaller or equal to that of $\b{K}_X$.
\end{remark}

\section{Fully Symmetric Bayes--Sard Cubature} \label{sec:bsc-fss}

In this section we first review the Bayes--Sard cubature method from~\cite{Karvonen2018c} and then derive a generalisation of Theorem~\ref{thm:bc-fss} for this method.

\subsection{Bayes--Sard Cubature} \label{sec:bsc}

In the standard Bayesian cubature method the mean function $m$ must be \emph{a priori} specified.
This requirement is relaxed in Bayes--Sard cubature \cite{Karvonen2018c}, where a hierarchical approach is taken instead.
Specifically, in Bayes--Sard cubature the prior mean function is given the parametric form
\begin{eqnarray*}
m_{\b{\theta}}(\b{x}) & = & \theta_1 \phi_1(\b{x}) + \cdots + \theta_Q \phi_Q(\b{x}) \\
& = & \b{\theta}^\T \b{\phi}(\b{x}),
\end{eqnarray*}
where $\b{\phi}(\b{x}) \in \mathbb{R}^Q$ has entries $[\b{\phi}(\b{x})]_i = \phi_i(\b{x})$ and the parameter vector $\b{\theta} = (\theta_1,\ldots,\theta_Q) \in \mathbb{R}^Q$ represents coefficients in a pre-defined basis consisting of functions $\phi_i : M \rightarrow \mathbb{R}$, $i = 1,\dots,Q$, that are assumed $\nu$-integrable and that span a finite-dimensional linear function space $\pi \coloneqq \text{span}(\phi_1,\dots,\phi_Q)$.
That is, $m_{\b{\theta}} \in \pi$ for any $\b{\theta} \in \mathbb{R}^Q$.
Then, for a positive-definite $\b{\Sigma} \in \mathbb{R}^{Q \times Q}$, a Gaussian hyper-prior distribution 
\begin{equation*}
\b{\theta} \sim \mathrm{N}(\b{0}, \b{\Sigma})
\end{equation*}
is specified.
The conditional distribution $f_N$ of this field, based as before on data $\mathcal{D}$, is again Gaussian.
In particular, when $\b{\Sigma}^{-1} \to \b{0}$ (meaning that the prior on $\b{\theta}$ becomes improper, or weakly informative\footnote{See~\cite{Larkin1974,OHagan1991} for slightly different earlier formulations where an improper prior is placed ``directly'' on $\b{\theta}$.}) and assuming that $Q \leq N$, the posterior mean and variance take the forms
\begin{align}
m_N(\b{x}) ={}& \b{\alpha}^\T \b{k}_X(\b{x}) + \b{\beta}^\T \b{\phi}(\b{x}), \label{eqn:BSC-GP-mean} \\
\begin{split}\label{eqn:BSC-GP-var}
k_N(\b{x},\b{x}') ={}& k(\b{x},\b{x}') - \b{k}_X(\b{x})^\T \b{K}_{X}^{-1} \b{k}_X(\b{x}') \\
& + [\b{\Phi}_X^\T \b{K}_{X}^{-1} \b{k}_X(\b{x}) - \b{\phi}(\b{x})]^\T [\b{\Phi}_X^\T \b{K}_{X}^{-1} \b{\Phi}_X]^{-1} \\
& \quad \times [\b{\Phi}_X^\T \b{K}_{X}^{-1} \b{k}_X(\b{x}') - \b{\phi}(\b{x}')],
\end{split}
\end{align}
where $\b{\Phi}_X \in \mathbb{R}^{N \times Q}$ with entries $[\b{\Phi}_X]_{i,j} = \phi_j(\b{x}_i)$ is called the \emph{Vandermonde matrix} and the vectors $\b{\alpha}$ and $\b{\beta}$ are defined via the linear system
\begin{equation}\label{eq: alpha-beta}
\begin{bmatrix} \b{K}_{X} & \b{\Phi}_X \\ \b{\Phi}_X^\T & \b{0} \end{bmatrix} \begin{bmatrix} \b{\alpha} \\ \b{\beta} \end{bmatrix} = \begin{bmatrix} \b{f}^\dagger_X \\ \b{0} \end{bmatrix}.
\end{equation}
For there to exist a unique solution to~\eqref{eq: alpha-beta}, the Vandermonde matrix has to be of full rank.
This technical condition, equivalent to the zero function being the only element of $\pi$ vanishing on $X$, is known as $\pi$-\emph{unisolvency} of the point set $X$.
Throughout the article we assume this is the case; see~\cite[Section~2.2]{Wendland2005} or~\cite[Supplement~B]{Karvonen2018c} for more information and examples of unisolvent point sets.

The output of the \emph{Bayes--Sard cubature} method is the posterior marginal distribution of the integral, namely
\begin{eqnarray} \label{eq:bsc-output}
\int_M f_N(\b{x}) \dif\nu(\b{x}) \sim \mathrm{N}\big( \mu_N(f^\dagger) , \sigma_N^2 \big).
\end{eqnarray}
The mean and variance, obtained by integrating~\eqref{eqn:BSC-GP-mean} and\eqref{eqn:BSC-GP-var}, are
\begin{align*}
\mu_N(f^\dagger) ={}& (\b{w}_{X}^k)^\T \b{f}^\dagger_X, \\ 
\begin{split}
\sigma_N^2 ={}& k_{\nu,\nu} - \b{k}_{\nu,X}^\T \b{K}_{X}^{-1} \b{k}_{\nu,X} \\
&+ (\b{w}_{X}^\pi)^\T \big( \b{\Phi}_X^\T \b{K}_{X}^{-1} \b{k}_{\nu,X} - \b{\phi}_\nu \big),
\end{split}
\end{align*}
where $\b{\phi}_\nu \in \mathbb{R}^Q$ has the entries $[\b{\phi}_\nu]_i = \int_M \phi_i(\b{x}) \mathrm{d}\nu(\b{x})$ and the weight vectors $\b{w}_{X}^k \in \mathbb{R}^N$ and $\b{w}_{X}^\pi \in \mathbb{R}^Q$ are the solution to the linear system
\begin{equation}\label{eqn:BSC-weights}
\begin{bmatrix} \b{K}_{X} & \b{\Phi}_X \\ \b{\Phi}_X^\T & \b{0} \end{bmatrix} \begin{bmatrix} \b{w}_{X}^k \\ \b{w}_{X}^\pi \end{bmatrix} = \begin{bmatrix} \b{k}_{\nu,X} \\ \b{\phi}_\nu \end{bmatrix}.
\end{equation}
The Bayes--Sard weights $\b{w}_X^k$, like the standard Bayesian cubature weights, have a worst case interpretation:
\begin{equation*}
\b{w}_X^k = \argmin_{\b{w} \in \mathbb{R}^N} \, \sup_{\norm[0]{f^\dagger}_{\mathcal{H}(k)} \leq 1} \, \abs[3]{\int_M f^\dagger \dif \nu - \sum_{i=1}^N w_i f^\dagger(\b{x}_i)}
\end{equation*}
subject to the linear constraints $\sum_{i=1}^N w_i \phi_j(\b{x}_i) = I(\phi_j)$ for $j = 1,\ldots,Q$~\cite{DeVore2017}.

The Bayes--Sard method has some important theoretical and practical advantages over the standard Bayesian cubature method, which motivate us to study it in detail:
\begin{itemize}
\item The posterior mean $\mu_N(f^\dagger)$ is exactly equal to the integral $I(f^\dagger)$ if $f^\dagger \in \pi$.
In particular, if $\pi$ contains a non-zero constant function then $\sum_{i=1}^N w_{X,i}^k = 1$ so that the cubature rule is \emph{normalised} \rev{(however, non-negativity of the weights is not guaranteed\footnote{\rev{It is possible to employ a positivity constraint~\cite{Ehler2018}, but in that case there is no convenient closed-form expression for the weights and the Bayesian interpretation is sacrificed.}})}.
This can improve the stability of the method in high-dimensional settings \cite{Karvonen2018c}.
In general, if $\pi$ is the set of polynomials up to a certain order $q$, then the posterior mean is recognised as a cubature rule of algebraic degree $q$~\cite[Definition~3.1]{Cools1997}.
\item Given any cubature rule $\mu(f^\dagger) = \sum_{i=1}^N w_i f^\dagger(\b{x}_i)$ for specified $w_i \in \mathbb{R}$ and $\b{x}_i \in M$, and given any covariance function $k$, one can find an $N$-dimensional function space $\pi$ such that $\mu_N = \mu$.
Furthermore, the posterior standard deviation $\sigma_N$ coincides with the worst case error of the cubature rule $\mu$ in the Hilbert space induced by $k$~\cite[Section~2.4]{Karvonen2018c}.
This demonstrates that any cubature rule can be interpreted as the posterior mean under an infinitude of prior models, providing a bridge between classical and Bayesian cubature methods.
\end{itemize}

The dimension of the linear system in \eqref{eqn:BSC-weights} is $N + Q$.
Thus the computational cost associated with the Bayes--Sard method is strictly greater than that of standard Bayesian cubature; at least $\mathcal{O}(N^3)$ in general. 
It is therefore of considerable practical interest to ask whether symmetry exploits can also be developed for the Bayes--Sard method.

\subsection{A Symmetry Exploit for Bayes--Sard Cubature}

In this section we present a novel result that enables fully symmetric sets to be exploited in the Bayes--Sard cubature method.
In what follows we only consider a function space $\pi$ spanned by even monomials exhibiting symmetries.\footnote{Odd monomials come for ``free''; see Remark~\ref{remark:odd-polynomials}.}
In practice, we do not believe this to be a significant restriction since polynomials typically serve as a good and functional default and, in fact, one retains considerable freedom in selecting the polynomials, not being restricted to, for example, spaces of all polynomials of at most a given degree.

Let $\pi_\alpha \subset \mathbb{N}_0^m$ denote a finite collection of multi-indices that in turn define the function space $\pi$:
\begin{equation*}
\pi = \mathrm{span} \{ \b{x}^{\b{\alpha}} \, \colon \, \b{\alpha} \in \pi_\alpha \}.
\end{equation*}
Here $\b{x}^{\b{\alpha}}$ denotes the monomial $x_1^{\alpha_1} \times \cdots \times x_m^{\alpha_m}$.
Define the index set
\begin{equation*}
\mathbb{E}_0^m \coloneqq \{ \b{\alpha} \in \mathbb{N}_0^m \, \colon \, \alpha_i \text{ is even for every } i=1,\ldots,m \}.
\end{equation*}
Our development will require that $\pi_\alpha$ is a union of $J_\alpha \in \mathbb{N}$ non-negative fully symmetric sets in $\mathbb{E}_0^m$.
That is, $\b{\alpha} \in \pi_\alpha$ implies $\b{P}\b{\alpha} \in \pi_\alpha$ for any permutation matrix $\b{P} \in \perm_m$ and there exist distinct $\b{\alpha}^1, \ldots, \b{\alpha}^{J_\alpha} \in \mathbb{E}_0^m$ such that
\begin{equation*}
\pi_\alpha = \bigcup_{j=1}^{J_\alpha} [\b{\alpha}^{j}]^+.
\end{equation*}
To prove a Bayes--Sard analogue of Theorem~\ref{thm:bc-fss}, we need four simple lemmas:

\begin{lemma} \label{lemma:integral-invariance} Suppose that $M$ and $\nu$ are each fully symmetric. 
If $\b{\alpha} \in \mathbb{E}_0^m$ then $I(\b{x}^{\b{\alpha}}) = I(\b{x}^{\b{P} \b{\alpha}})$ for any $\b{P} \in \perm_m$. 
\end{lemma}
\begin{proof} First, observe that $(\b{P}^{-1} \b{x})^{\b{\alpha}} = \b{x}^{\b{P} \b{\alpha}}$. \rev{By the change of variables formula of pushforwards and the assumption $\b{P}^{-1}_*(\nu) = \nu$,}
\begin{equation*}
  \begin{split}
  I(\b{x}^{\b{\alpha}}) &= \int_M \b{x}^{\b{\alpha}} \dif \nu(\b{x}) \\
  &= \int_M \b{x}^{\b{\alpha}} \dif \b{P}^{-1}_*(\nu)(\b{x}) \\
  &= \int_M (\b{P}^{-1} \b{x})^{\b{\alpha}} \dif \nu(\b{x}) \\
 &= \int_M \b{x}^{\b{P} \b{\alpha}} \dif \nu(\b{x}) \\
  &= I(\b{x}^{\b{P} \b{\alpha}})
\end{split}
\end{equation*}
for any $\b{\alpha} \in \mathbb{N}_0^m$. \qed
\end{proof}

\begin{lemma} \label{lemma:integral-invariance-kernel}
Suppose that $M$, $\nu$, and $k$ are each fully symmetric and let $\b{\lambda} \in M$.
Then $k_\nu(\b{x}) = k_\nu(\b{\lambda})$ for every $\b{x} \in [\b{\lambda}]$.
\end{lemma}
\begin{proof}
  \rev{The proof is essentially identical to that of Lemma~\ref{lemma:integral-invariance}.} \qed
\end{proof}

\begin{lemma}\label{lemma:sum-invariance} Let $\b{\lambda} \in \mathbb{R}^m$ and $\b{\alpha} \in \mathbb{E}_0^m$. Then
\begin{align}
\sum_{ \b{\beta} \in [\b{\alpha}]^+} \b{x}^{\b{\beta}} &= \sum_{ \b{\beta} \in [\b{\alpha}]^+} \b{\lambda}^{\b{\beta}} &&\text{ for any } \quad \b{x} \in [\b{\lambda}], \label{eq:xpowb-invariance} \\
\sum_{\b{x} \in [\b{\lambda}]} \b{x}^{\b{\beta}} &= \sum_{\b{x} \in [\b{\lambda}]} \b{x}^{\b{\alpha}} &&\text{ for any } \quad \b{\beta} \in [\b{\alpha}]^+. \label{eq:xpowb2-invariance}
\end{align}
\end{lemma}

\begin{proof}
For any $\b{\alpha} \in \mathbb{E}_0^m$, $\b{x} \in [\b{\lambda}]$, and $\b{P} \in \permsc_m$,
\begin{equation*}
\sum_{ \b{\beta} \in [\b{\alpha}]^+} \b{x}^{\b{\beta}} = \sum_{ \b{\beta} \in [\b{\alpha}]^+} (\b{P}^{-1} \b{P}\b{x})^{\b{\beta}} = \sum_{ \b{\beta} \in [\b{\alpha}]^+} (\b{P} \b{x})^{\b{P}^+ \b{\beta}},
\end{equation*}
where $\b{P}^+ \in \perm_m$ has the elements $[\b{P}^+]_{ij} = \abs[0]{[\b{P}]_{ij}}$ and the second equality follows from the fact that every element of $\b{\beta}$ is even. 
Because $[\b{P}^+ \b{\alpha}]^+ = [\b{\alpha}]^+$, it follows that $\sum_{ \b{\beta} \in [\b{\alpha}]^+} \b{x}^{\b{\beta}} = \sum_{ \b{\beta} \in [\b{\alpha}]^+} (\b{P} \b{x})^{\b{\beta}}$. That is, 
\begin{equation*}
\sum_{ \b{\beta} \in [\b{\alpha}]^+} \b{x}^{\b{\beta}} = \sum_{ \b{\beta} \in [\b{\alpha}]^+} \b{\lambda}^{\b{\beta}}
\end{equation*}
since $\b{\lambda} = \b{P}\b{x}$ for some $\b{P} \in \permsc_m$.
Consider then the ``transpose'' sum $\sum_{\b{x} \in [\b{\lambda}]} \b{x}^{\b{\beta}}$ for $\b{\beta} \in [\b{\alpha}]^+$. Similar arguments as above establish that
\begin{equation*}
\sum_{\b{x} \in [\b{\lambda}]} \b{x}^{\b{\beta}} = \sum_{\b{x} \in [\b{\lambda}]} (\b{P} \b{x})^{\b{P} \b{\beta}} = \sum_{\b{x} \in [\b{\lambda}]} \b{x}^{\b{P} \b{\beta}}
\end{equation*}
for any $\b{P} \in \perm_m$. Consequently,
\begin{equation*}
\sum_{\b{x} \in [\b{\lambda}]} \b{x}^{\b{\beta}} = \sum_{\b{x} \in [\b{\lambda}]} \b{x}^{\b{\alpha}}
\end{equation*}
for every $\b{\beta} \in [\b{\alpha}]^+$. \qed
\end{proof}

\begin{lemma} \label{lemma: fss-cross-sets}
Let $\b{\lambda}, \b{\lambda}' \in \mathbb{R}^m$ and suppose that the kernel $k$ is fully symmetric.
Then
\begin{equation*}
\sum_{ \b{x}' \in [\b{\lambda}']} k(\b{x},\b{x}') = \sum_{ \b{x}' \in [\b{\lambda}']} k(\b{\lambda},\b{x}') \quad \text{for any} \quad \b{x} \in [\b{\lambda}].
\end{equation*}
\end{lemma}

\begin{proof}
For any $\b{x} \in [\b{\lambda}]$ there is $\b{P}_{\b{x}} \in \permsc_m$ such that $\b{x} = \b{P}_{\b{x}} \b{\lambda}$.
Therefore
\begin{equation*}
\begin{split}
\sum_{\b{x}' \in [\b{\lambda}']} k(\b{x},\b{x}') &= \sum_{\b{x}' \in [\b{\lambda}']} k(\b{P}_{\b{x}} \b{\lambda},\b{x}') \\
&= \sum_{\b{x}' \in [\b{\lambda}']} k(\b{P}_{\b{x}}^{-1} \b{P}_{\b{x}} \b{\lambda}, \b{P}_{\b{x}}^{-1} \b{x}') \\
&= \sum_{\b{x}' \in [\b{\lambda}']} k(\b{\lambda}, \b{P}_{\b{x}}^{-1} \b{x}') \\
&= \sum_{\b{x}' \in [\b{P}_{\b{x}}^{-1} \b{\lambda}']} k(\b{\lambda}, \b{x}'),
\end{split}
\end{equation*}
and the claim follows from the fact that $[\b{P} \b{\lambda}'] = [\b{\lambda}']$ for any $\b{P} \in \permsc_m$. \qed
\end{proof}

We are now ready to prove the main result of this section.
Theorem \ref{thm:BSC-FSS} establishes sufficient conditions for the Bayes--Sard cubature rule to be fully symmetric and, in that case, provides an explicit simplification of the its output~\eqref{eq:bsc-output}.

\begin{theorem}\label{thm:BSC-FSS}
Consider the Bayes--Sard cubature method based on a domain $M$, measure $\nu$, and kernel $k$ that are each fully symmetric. 
Suppose that 
\begin{equation*}
\pi = \mathrm{span} \{ \b{x}^{\b{\alpha}} \, \colon \, \b{\alpha} \in \pi_\alpha \} \quad \text{ and } \quad \pi_\alpha = \bigcup_{j=1}^{J_\alpha} [\b{\alpha}^{j}]^+ 
\end{equation*}
for a collection $\mathcal{A} = \{ \b{\alpha}^1,\ldots,\b{\alpha}^{J_\alpha} \} \subset \mathbb{E}_0^m$ of distinct even multi-indices and that $X$ is a union of $J$ distinct fully symmetric sets: $X = \bigcup_{j=1}^J [\b{\lambda}^j]$ for a collection $\Lambda = \{\b{\lambda}^1,\ldots,\b{\lambda}^J\} \subset M$ of distinct generator vectors.
Then the output of the Bayes--Sard cubature method can be expressed in the fully symmetric form
\begin{align*}
\mu_N(f^\dagger) ={}& \sum_{j=1}^J w_{\Lambda,j}^k \sum_{\b{x} \in [\b{\lambda}^j]} f^\dagger(\b{x}) ,\\
\sigma_N^2 ={}& k_{\nu,\nu} - \sum_{j=1}^J w_{\Lambda,j}^\sigma k_\nu(\b{\lambda}^j) n_j^\Lambda \\
&+ \sum_{j=1}^{J_\alpha} w_{\mathcal{A},j}^\pi n_j^\mathcal{A} \Bigg[ \sum_{i=1}^J w_{\Lambda,i}^\sigma \sum_{\b{x} \in [\b{\lambda}^{i}]} \b{x}^{\b{\alpha}^j} - I(\b{x}^{\b{\alpha}^j}) \Bigg],
\end{align*}
where $n_j^\Lambda = \#[\b{\lambda}^j]$, $n_j^\mathcal{A} = \#[\b{\alpha}^j]^+$, and $\b{w}_\Lambda^\sigma \in \mathbb{R}^J$ are the weights $\b{w}_\Lambda$ in Theorem~\ref{thm:bc-fss}.
The weights \sloppy{${\b{w}_\Lambda^k \in \mathbb{R}^J}$} and $\b{w}_\mathcal{A}^\pi \in \mathbb{R}^{J_\alpha}$ form the solution to the linear system
\begin{equation}\label{eqn:block-system-fss}
\begin{bmatrix} \b{S} & \b{A} \\ \b{B} & \b{0} \end{bmatrix} \begin{bmatrix} \b{w}_\Lambda^k \\ \b{w}_\mathcal{A}^\pi \end{bmatrix} = \begin{bmatrix} \b{k}_{\nu,\Lambda} \\ \b{\phi}_{\nu,\mathcal{A}} \end{bmatrix}
\end{equation}
of $J + J_\alpha$ equations, where $[\b{k}_{\nu,\Lambda}]_j = k_\nu(\b{\lambda}^j)$, $[\b{\phi}_{\nu,\mathcal{A}}]_j = I(\b{x}^{\b{\alpha}^j})$, $[\b{S}]_{ij} = \sum_{\b{x} \in [\b{\lambda}^j]} k(\b{\lambda}^i,\b{x})$, $[\b{A}]_{ij} = \sum_{\b{\beta} \in [\b{\alpha}^j]^+} (\b{\lambda}^i)^{\b{\beta}}$, and $[\b{B}]_{ij} = \sum_{\b{x} \in [\b{\lambda}^j]} \b{x}^{\b{\alpha}^i}$.

\end{theorem}

\begin{proof} 
The linear system~\eqref{eqn:block-system-fss} is equivalent to
\begin{equation*}
\sum_{j=1}^J w_{\Lambda,j}^k S_{ij} + \sum_{j=1}^{J_\alpha} w_{\mathcal{A},j}^\pi A_{ij} = k_\nu(\b{\lambda}^i) \: \text{ for } \: i \in \{1,\ldots,J\}
\end{equation*}
and
\begin{equation*}
\sum_{j=1}^J w_{\Lambda,j}^k B_{ij} = I(\b{x}^{\b{\alpha}^i}) \quad \text{ for } \quad i \in \{1, \ldots, J_\alpha\}.
\end{equation*}
These two groups of equations are equivalent, respectively, to the $N$ equations (Lemmas~\ref{lemma:integral-invariance-kernel} and~\ref{lemma: fss-cross-sets} and~\eqref{eq:xpowb-invariance})
\begin{equation*}
\sum_{j=1}^J w_{\Lambda,j}^k \sum_{\b{x}' \in [\b{\lambda}^j]} k(\b{x},\b{x}') + \sum_{j=1}^{J_\alpha} w_{\mathcal{A},j}^\pi \sum_{\b{\beta} \in [\b{\alpha}^j]^+} \b{x}^{\b{\alpha}} = k_\nu(\b{x})
\end{equation*}
for $i \in \{1,\ldots,J\}$, $\b{x} \in [\b{\lambda}^i]$, and to the $Q$ equations (Lemma~\ref{lemma:integral-invariance} and~\eqref{eq:xpowb2-invariance})
\begin{equation*}
\sum_{j=1}^J w_{\Lambda,j}^k \sum_{\b{x}' \in [\b{\lambda}^j]} (\b{x}')^{\b{\alpha}} = I(\b{x}^{\b{\alpha}})
\end{equation*}
for $\b{\alpha} \in \pi_\alpha$.
From these two equations we recognise that 
\begin{equation*}
\b{w}_X^k = \begin{bmatrix} w_{\Lambda,1}^k \b{1}_{n_1^\Lambda} \\ \vdots \\ w_{\Lambda,J}^k \b{1}_{n_J^\Lambda} \end{bmatrix}
\quad \text{and} \quad 
\b{w}_X^\pi = \begin{bmatrix} w_{\mathcal{A},1}^k \b{1}_{n_1^\mathcal{A}} \\ \vdots \\ w_{\mathcal{A},J_\alpha}^k \b{1}_{n_{J_\alpha}^\mathcal{A}} \end{bmatrix}
\end{equation*}
solve the full Bayes--Sard weight system
\begin{equation*}
\begin{bmatrix} \b{K}_X & \b{\Phi}_X \\ \b{\Phi}_X^\T & \b{0} \end{bmatrix} \begin{bmatrix} \b{w}_X^k \\ \b{w}_X^\pi \end{bmatrix} = \begin{bmatrix} \b{k}_{\nu,X} \\ \b{\phi}_{\nu} \end{bmatrix}.
\end{equation*}
The expression for the Bayes--Sard variance $\sigma_N^2$ can be obtained by first recognising that the unique elements of $\b{K}_X^{-1} \b{k}_{\nu,X}$ are precisely the weights $\b{w}_\Lambda$ in Theorem~\ref{thm:bc-fss}, here denoted $\b{w}_\Lambda^\sigma$. 
Then we compute
\begin{equation*}
\b{k}_{\nu,X}^\T \b{K}_X^{-1} \b{k}_{\nu,X} = \sum_{j=1}^J w_{\Lambda,j}^\sigma k_\nu(\b{\lambda}^j) n_j^\Lambda
\end{equation*}
and
\begin{equation*}
\begin{split}
(&\b{w}_{X}^\pi)^\T \big( \b{\Phi}_X^\T \b{K}_{X}^{-1} \b{k}_{\nu,X} - \b{\phi}_\nu \big) \\
&= (\b{w}_{X}^\pi)^\T \begin{bmatrix} \Big( \sum_{j=1}^J w_{\Lambda_j}^\sigma \sum_{\b{x} \in [\b{\lambda}]^j} \b{x}^{\b{\alpha}^1} - I(\b{x}^{\b{\alpha}^1}) \Big) \b{1}_{n_1^\mathcal{A}} \\ \vdots \\ \Big( \sum_{j=1}^J w_{\Lambda_j}^\sigma \sum_{\b{x} \in [\b{\lambda}]^j} \b{x}^{\b{\alpha}^{J_\alpha}} - I(\b{x}^{\b{\alpha}^{J_\alpha}}) \Big) \b{1}_{n_{J_\alpha}^\mathcal{A}} \end{bmatrix}
\end{split}
\end{equation*}
that, when expanded, yields the result. \qed
\end{proof}

\begin{remark}\label{remark:odd-polynomials}
The polynomial space $\pi$ could be appended with fully symmetric collections of odd polynomials (i.e., by using additional basis functions $\b{x}^{\b{\beta}}$, $\b{\beta} \in [\b{\alpha}]^+$ for $\b{\alpha} \notin \mathbb{E}_0^m$). 
However, by doing this one gains nothing since the weights in $\b{w}_{\mathcal{A}}^\pi$ corresponding to these basis functions turn out to be zero.
This is quite easy to see from the easily proven facts that $\sum_{\b{x} \in [\b{\lambda}]} \b{x}^{\b{\beta}} = 0$ and $I(\b{x}^{\b{\beta}}) = 0$ whenever $\b{\beta} \notin \mathbb{E}_0^m$.
\end{remark}

Just like Theorem~\ref{thm:bc-fss} for the standard Bayesian cubature, Theorem~\ref{thm:BSC-FSS} reduces the number of kernel and basis function evaluations from roughly $N^2 + Q^2$ to $NJ + NJ_\alpha$ and the size of the linear system that needs to be solved from $N+Q$ to $J+J_\alpha$.
Typically, this translates to a significant computational speed-up; see Section~\ref{sec:zcb} for a numerical example involving point sets of up to $N = 179,\!400$.
Such results could not realistically be obtained by direct solution of the original linear system~\eqref{eqn:BSC-weights}.

\section{Fully Symmetric Multi-Output Bayesian Cubature} \label{subsec: vvbc}

In this section we review the multi-output Bayesian cubature method recently proposed by Xi et al.\@~\cite{Xi2018} and show how to exploit fully symmetric sets in reducing computational complexity of this method.

\subsection{Multi-Output Bayesian Cubature} \label{subsec: vvbc-fss}

One often needs to integrate a number of related integrands, $f_1^\dagger,\ldots,f_D^\dagger \colon  M \to \mathbb{R}$.
It is of course trivial to treat these as a set of $D$ independent integrals and apply either the standard Bayesian or Bayes--Sard cubature method to approximate each integral.
However, in many cases the relationship between the integrands can be explicitly modelled and leveraged.

Such a setting can be handled by modelling a single vector-valued function $\b{f}^\dagger \coloneqq (f_1^\dagger,\ldots,f_D^\dagger) \colon M \to \mathbb{R}^D$ as a vector-valued Gaussian field; full details can be found in~\cite{Alvarez2012}.
In this case, the data $\mathcal{D}$ consist of evaluations
\begin{equation*}
\b{f}^\dagger_{d,X_d} = \big[ f^\dagger_d(\b{x}_{d1}),\ldots,f^\dagger_d(\b{x}_{dN}) \big]^\T \in \mathbb{R}^N
\end{equation*}
at points $X_d = \{ \b{x}_{d1},\ldots,\b{x}_{dN}\} \subset M$ for each \sloppy{${d = 1,\ldots,D}$}.
In this section we denote $X = \{ X_d \}_{d=1}^D$.
The assumption that each integrand is evaluated at $N$ points is made only for notational simplicity; all results can be easily modified to accommodate different numbers of points for each integrand.
Evaluations of each integrand are concatenated into the vector
\begin{equation*}
\begin{split}
\b{f}_X^\dagger &= \big[ f^\dagger_1(\b{x}_{11}), \ldots, f^\dagger_1(\b{x}_{1N}), \ldots, \\
& \hspace{1cm} f^\dagger_D(\b{x}_{D1}), \ldots, f^\dagger_D(\b{x}_{DN}) \big]^\T \in \mathbb{R}^{DN}.
\end{split}
\end{equation*}

In multi-output Bayesian cubature the integrand is modelled as a vector-valued Gaussian field $\b{f} \in \mathbb{R}^D$ characterised by vector-valued mean function $\b{m} \colon M \rightarrow \mathbb{R}^D$ and matrix-valued covariance function \sloppy{${\b{k} \colon M \times M \rightarrow \mathbb{R}^{D \times D}}$}. For notational simplicity, the prior mean function is fixed at $\b{m} \equiv \b{0}$.
The conditional distribution $\b{f}_N$ of this field, based on the data $\mathcal{D} = (X, \b{f}_X^\dagger)$, is also Gaussian with mean and covariance functions
\begin{eqnarray*}
\b{m}_N(\b{x}) & = & \b{k}_X(\b{x})^\T \b{K}_X^{-1} \b{f}^\dagger_X, \\
\b{k}_N(\b{x},\b{x}') & = & \b{k}(\b{x},\b{x}') - \b{k}_X(\b{x})^\T \b{K}_X^{-1} \b{k}_X(\b{x}).
\end{eqnarray*}
Here, in contrast to~\ref{eqn:GP-mean} and~\ref{eqn:GP-var}, all objects are of extended dimensions:
\begin{align*}
\b{k}_X(\b{x}) &= \begin{bmatrix} \b{k}_{X_1}(\b{x}) \\ \vdots \\ \b{k}_{X_D}(\b{x}) \end{bmatrix} \in \mathbb{R}^{DN \times D}, \\
\b{K}_X &= \begin{bmatrix} \b{K}_{X_1,X_1}^{11} & \cdots & \b{K}_{X_1,X_D}^{1D} \\ \vdots & \ddots & \vdots \\ \b{K}_{X_D,X_1}^{D1} & \cdots & \b{K}_{X_D,X_D}^{DD} \end{bmatrix} \in \mathbb{R}^{DN \times DN},
\end{align*}
where $\b{k}_{X_d}(\b{x})$ and $\b{K}_{X_d,X_q}^{dq}$ are the $N \times D$ and $N \times N$ matrices
\begin{align*}
\b{k}_{X_d}(\b{x}) &= \begin{bmatrix} [\b{k}(\b{x}_{d1},\b{x})]_{11} & \cdots & [\b{k}(\b{x}_{d1},\b{x})]_{1D} \\ \vdots & \ddots & \vdots \\ [\b{k}(\b{x}_{dN},\b{x})]_{11} & \cdots & [\b{k}(\b{x}_{dN},\b{x})]_{DD} \end{bmatrix}, \\
\b{K}_{X_d,X_q}^{dq} &= \begin{bmatrix} [\b{k}(\b{x}_{d1},\b{x}_{q1})]_{dq} & \cdots & [\b{k}(\b{x}_{d1},\b{x}_{qN})]_{dq} \\ \vdots & \ddots & \vdots \\ [\b{k}(\b{x}_{dN},\b{x}_{q1})]_{dq} & \cdots & [\b{k}(\b{x}_{dN},\b{x}_{qN})]_{dq} \end{bmatrix}.
\end{align*}
The output of the \emph{multi-output} (or \emph{vector-valued}) Bayesian cubature method is a $D$-dimensional Gaussian random vector:
\begin{eqnarray*}
\int_M \b{f}_N(\b{x}) \mathrm{d}\nu(\b{x}) \sim \mathrm{N}\big( \b{\mu}_N(\b{f}^\dagger) , \b{\Sigma}_N \big)
\end{eqnarray*}
with
\begin{eqnarray}
\b{\mu}_N(\b{f}^\dagger) & = & \b{k}_{\nu,X}^\T \b{K}_X^{-1} \b{f}^\dagger_X, \label{eq: multi output 1} \\
\b{\Sigma}_N & = & \b{k}_{\nu,\nu} - \b{k}_{\nu,X}^\T \b{K}_X^{-1} \b{k}_{\nu,X}, \label{eq: multi output 2}
\end{eqnarray}
where $\b{k}_{\nu,X} = \int_{M} \b{k}_X(\b{x}) \dif \nu(\b{x}) \in \mathbb{R}^{DN \times D}$ and \sloppy{${\b{k}_{\nu,\nu} = \int_M \b{k}(\b{x},\b{x}') \dif \nu(\b{x}) \dif \nu(\b{x}') \in \mathbb{R}^{D \times D}}$}.
Equivalently, the posterior mean and variance can be written in terms of the weights
\begin{equation} \label{eq:mo-bc-weights}
\b{W}_X = \b{K}_X^{-1} \b{k}_{\nu,X} = [\b{W}_1^\T \cdots \b{W}_D^\T]^\T \in \mathbb{R}^{DN \times D},
\end{equation}
where $\b{W}_d \in \mathbb{R}^{N \times D}$.
For example, mean of the $d$th integral then takes the form
\begin{equation}
\mu_N(f_d^\dagger) = [\b{\mu}_N(\b{f}^\dagger)]_d = \sum_{q=1}^D \sum_{i=1}^N [\b{W}_q]_{id} f^\dagger_q(\b{x}_{qi}). \label{eq: dth integral multi output}
\end{equation}
If the $d$th integrand is modelled as independent of all the other integrands, the posterior mean~\eqref{eq: dth integral multi output} reduces to the standard Bayesian cubature posterior mean~\eqref{eq: BC mean}.

\subsection{Separable Kernels}

The structure of matrices appearing in the multi-output Bayesian cubature equations can be simplified when the multi-output kernel is \emph{separable}.
This means that there is a positive-definite $\b{B} \in \mathbb{R}^{D \times D}$ such that
\begin{equation}
\b{k}(\b{x},\b{x}') = \b{B} c(\b{x},\b{x}') \label{eq: sep kernel}
\end{equation}
for some positive-definite kernel $c \colon M \times M \to \mathbb{R}$.
The matrices $\b{K}_X$ and $\b{k}_{\nu,X}$ now assume the simplified forms
\begin{align*}
\b{k}_{\nu,X} &= \begin{bmatrix} B_{11} \b{c}_{\nu,X_1} & \cdots & B_{1D} \b{c}_{\nu,X_D} \\ \vdots & \ddots & \vdots \\ B_{D1} \b{c}_{\nu,X_1} & \cdots & B_{DD} \b{c}_{\nu,X_D} \end{bmatrix}, \\
\b{K}_X &= \begin{bmatrix} B_{11} \b{C}_{X_1,X_1} & \cdots & B_{1D} \b{C}_{X_1,X_D} \\ \vdots & \ddots & \vdots \\ B_{D1} \b{C}_{X_D,X_1} & \cdots & B_{DD} \b{C}_{X_D,X_D} \end{bmatrix},
\end{align*}
where $[\b{c}_{\nu,X_d}]_i = c_{\nu}(\b{x}_{di})$ and $[\b{C}_{X_d,X_q}]_{ij} = c(\b{x}_{di},\b{x}_{qj})$.
However, even with the simplified structure afforded by the use of separable kernels, the implementation of multi-output Bayesian cubature remains computationally challenging, calling for some $(DN)^2$ kernel evaluations and solution to a linear system of dimension $DN$. 
This is problematic if a large number of integrands is to be handled simultaneously.
The next section demonstrates how fully symmetric points sets can be exploited to reduce this cost.

\begin{remark}
Note that using the same point set $X'$ for each integrand yields immediate computational simplification, since in this case the above matrices can be written as Kronecker products:
\begin{equation*}
\b{k}_{\nu,X} = \b{B} \otimes \b{c}_{\nu,X'} \quad \text{and} \quad \b{K}_X = \b{B} \otimes \b{C}_{X',X'}.
\end{equation*}
However, this case is of little practical interest because, by the properties of the Kronecker product,
\begin{equation*}
\b{W}_X = \b{I}_D \otimes \b{w}_{X'},
\end{equation*}
where $\b{w}_{X'} \in \mathbb{R}^N$ are the standard Bayesian cubature weights~\eqref{eqn:BC-weights} for the covariance function $c$ and points $X'$~\cite[Supplements B and C.1]{Xi2018}.
That is, the integral estimates $\b{\mu}_N(\b{f}^\dagger)$ reduce to those given by the standard Bayesian cubature method applied independently to each integral.
\end{remark}

\subsection{A Symmetry Exploit for Multi-Output Bayesian Cubature}

Our main result in this section is a second generalisation of Theorem~\ref{thm:bc-fss}, in this case for the multi-output Bayesian cubature method.

\begin{theorem} \label{thm:mobc-fss}
Consider the multi-output Bayesian cubature method based on a separable matrix-valued kernel $\b{k}$.
Let the domain $M$, measure $\nu$, and uni-output kernel $c$ each be fully symmetric and fix the mean function to be $\b{m} \equiv \b{0}$. 
Suppose that each $X_d$ is a union of $J$ fully symmetric sets: $X_d = \bigcup_{j=1}^J [\b{\lambda}^{dj}]$ for some $\Lambda_d = \{ \b{\lambda}^{d1},\ldots,\b{\lambda}^{dJ} \} \subset M$ such that $n_j^\Lambda = \# [\b{\lambda}^{dj}]$ does not depend on $d$ and, consequently, $\# X_d = N$ for each \sloppy{${d=1,\ldots,D}$}.
Then the output of the multi-output Bayesian cubature method can be expressed in the fully symmetric form
\begin{align}
[\b{\mu}_N(\b{f}^\dagger)]_d &= \sum_{q=1}^D \sum_{j=1}^J [\b{W}_{\!\!\Lambda,q}]_{jd} \sum_{\b{x} \in [\b{\lambda}^{qj}]} f_q^\dagger(\b{x}), \label{eq:BSC-FSS-mean} \\
\b{\Sigma}_N &= \b{B} c_{\nu,\nu} - \sum_{d=1}^D \sum_{j=1}^J \b{P}_{dj} \b{B}_j^\text{\tiny{diag}},
\end{align}
where $\b{B}_j^\text{\tiny{diag}}$ is the diagonal $D \times D$ matrix formed out of the $j$th row of $\b{B}$ and
\begin{equation*}
\b{P}_{dj} = \begin{bmatrix} [\b{W}_{\!\!\Lambda,d}]_{j1} c_\nu(\b{\lambda}^{1j}) n_j^\Lambda & \cdots & [\b{W}_{\!\!\Lambda,d}]_{j1} c_\nu(\b{\lambda}^{Dj}) n_j^\Lambda \\ \vdots & \ddots & \vdots \\ [\b{W}_{\!\!\Lambda,d}]_{jD} c_\nu(\b{\lambda}^{1j}) n_j^\Lambda & \cdots & [\b{W}_{\!\!\Lambda,d}]_{jD} c_\nu(\b{\lambda}^{Dj}) n_j^\Lambda \end{bmatrix}.
\end{equation*}
The weight matrix
\begin{equation*}
\b{W}_{\!\!\Lambda} = [ \b{W}_{\!\!\Lambda,1}^\T \cdots \b{W}_{\!\!\Lambda,D}^\T ]^\T \in \mathbb{R}^{DJ \times D}, \quad \b{W}_{\!\!\Lambda,d} \in \mathbb{R}^{J \times D},
\end{equation*}
is the solution to the linear system $\b{S} \b{W}_{\!\!\Lambda} = \b{k}_{\nu,\Lambda}$, where
\begin{align*}
\b{S} &= \begin{bmatrix} B_{11} \b{S}_{11} & \cdots & B_{1D} \b{S}_{1D} \\ \vdots & \ddots & \vdots \\ B_{D1} \b{S}_{D1} & \cdots & B_{DD} \b{S}_{DD} \end{bmatrix} \in \mathbb{R}^{DJ \times DJ}, \\
[\b{S}_{dq}]_{ij} &= \sum_{\b{x} \in [\b{\lambda}^{qj}]} c( \b{\lambda}^{di}, \b{x} ), \\
\b{k}_{\nu,\lambda} &= \begin{bmatrix} B_{11} \b{c}_{\nu,\Lambda_1} & \cdots & B_{1D} \b{c}_{\nu,\Lambda_D} \\ \vdots & \ddots & \vdots \\ B_{D1} \b{c}_{\nu,\Lambda_1} & \cdots & B_{DD} \b{c}_{\nu,\Lambda_D} \end{bmatrix} \in \mathbb{R}^{DJ \times D}, \\
[\b{c}_{\nu,\Lambda_d}]_j &= c_\nu(\b{\lambda}^{dj}).
\end{align*} 
\end{theorem}

\begin{proof} 
The matrix equation $\b{S} \b{W}_{\!\!\Lambda} = \b{k}_{\nu,\Lambda}$ corresponds to the $D^2 J$ equations
\begin{equation*}
\sum_{q=1}^D B_{dq} \sum_{i=1}^J [\b{S}_{dq}]_{ji} [\b{W}_{\!\!\Lambda,q}]_{id'} = B_{dd'} c_{\nu}(\b{\lambda}^{d'j})
\end{equation*}
for $(d,d',j) \in \{1, \ldots, d\}^2 \times \{1,\ldots,J\}$.
In turn, through Lemmas~\ref{lemma:integral-invariance-kernel} and~\ref{lemma: fss-cross-sets}, these are equivalent to
\begin{equation*}
\sum_{q=1}^D B_{dq} \sum_{i=1}^J [\b{W}_{\!\!\Lambda,q}]_{id'} \sum_{\b{x} \in [\b{\lambda}^{qi}]} c( \b{x}_{dj'}, \b{x} ) = B_{dd'} c_{\nu}(\b{x}_{d'j'})
\end{equation*}
for $(d,d') \in \{1,\ldots,D\}^2$ and $j'=1,\ldots,n_j^\Lambda$, $j=1,\ldots,J$.
There are a total of $D^2 \sum_{j=1}^J n_j^\Lambda = D^2 N$ of these equations.
The weights 
\begin{equation*}
\b{W}_d = \begin{bmatrix} [\b{W}_{\Lambda,d}]_{11} \b{1}_{n^\Lambda_1} & \cdots & [\b{W}_{\Lambda,d}]_{1D} \b{1}_{n^\Lambda_1} \\ \vdots & \ddots & \vdots \\ [\b{W}_{\Lambda,d}]_{J1} \b{1}_{n^\Lambda_J} & \cdots & [\b{W}_{\Lambda,d}]_{JD} \b{1}_{n^\Lambda_J}\end{bmatrix} 
\end{equation*}
in \eqref{eq:mo-bc-weights} are then seen to solve the full matrix equation $\b{K}_X \b{W} = \b{k}_{\nu,X}$.
The expressions for the posterior mean and variance follow from straightforward manipulation of~\eqref{eq: multi output 1} and~\eqref{eq: multi output 2}.
\qed
\end{proof}

The computational complexity of forming the fully symmetric weight matrix $\b{W}_{\!\!\Lambda}$ is dominated by the $DJN$ kernel evaluations needed to form $\b{S}$ and the inversion of this $DJ \times DJ$ matrix. 
Due to $J$ often being orders of magnitude smaller than $N$, these tasks remain feasible even for a very large total number of points $DN$.
For example, in Section~\ref{sec:global-illu-int} the result of Theorem~\ref{thm:mobc-fss} is applied to facilitate the simultaneous computation of up to $D = 50$ integrals arising in a global illumination problem, each integrand being evaluated at up to $N = 288$ points.
Such results can barely be obtained by direct solution of the original linear system in~\eqref{eq:mo-bc-weights}.

\section{Symmetric Change of Measure} \label{sec:is-trick}

The results presented in this article, and those originally described in~\cite{Karvonen2018}, rely on the assumption that the measure $\nu$ is fully symmetric (see Section~\ref{sec:fss-objects}).
This is a strong restriction; most measures are not fully symmetric.
However, this assumption can be avoided in a relatively straightforward manner, which is now described.

Suppose that $M$ is a fully symmetric domain and that $\nu$ is an arbitrary measure, admitting a density $p_\nu$, against which the function $f^\dagger \colon M \to \mathbb{R}$ is to be integrated. 
Further suppose that there is a fully symmetric measure $\nu_*$ on $M$ such that $\nu$ is absolutely continuous with respect to $\nu_*$, and therefore admits a density $p_{\nu_*}$ such that the Radon--Nikodym derivative $\mathrm{d} \nu / \mathrm{d} \nu_* = p_\nu(\b{x}) / p_{\nu_*}(\b{x})$ is well-defined.
Then the integral of interest can be re-written as an integral with respect to the fully symmetric measure $\nu_*$:
\begin{eqnarray}\label{eq:is-trick}
\int_M f^\dagger(\b{x}) \dif \nu(\b{x}) & = & \int_{M} f^\dagger_*(\b{x}) \dif \nu_*(\b{x}) , \\ 
f^\dagger_*(\b{x}) & := & f^\dagger(\b{x}) \frac{p_\nu(\b{x})}{p_{\nu_*}(\b{x})}. \nonumber
\end{eqnarray}
Note that the existence of the second integral follows from the Radon--Nikodym property and the monotone convergence theorem.
Thus, the assumption of a fully symmetric measure $\nu$ in the statement of Theorems~\ref{thm:bc-fss},~\ref{thm:BSC-FSS}, and~\ref{thm:mobc-fss} is not overly restrictive.
This \emph{symmetric change of measure} technique is demonstrated on a numerical example in Section~\ref{sec: importance sampling example}.

\begin{remark}
Note that the situation here is unlike standard importance sampling (see e.g. Section~3.3 of \cite{Robert2013}), in that the importance distribution $\nu_*$ is required to be fully symmetric.
As such, it seems not obvious how to mathematically characterise an ``optimal'' choice of $\nu_*$.
Indeed, any notion of optimality ought also depend on the cubature method that will be used.
Nevertheless, obvious constructions (e.g. the choice of $\nu_*$ as an isotropic centred Gaussian for $\nu$ sub-Gaussian and $M = \mathbb{R}^m$) can work rather well.
\end{remark}

\section{Results} \label{sec:results}

In this section we assess the performance of the fully symmetric Bayes--Sard and fully symmetric multi-output Bayesian cubature methods based on computational simplifications provided in Theorems~\ref{thm:BSC-FSS} and~\ref{thm:mobc-fss}.
\rev{MATLAB} code for all examples is provided at \url{https://github.com/tskarvone/bc-symmetry-exploits}.

\subsection{\rev{Selection of Fully Symmetric Sets}}

\begin{figure*}[th!]
\centering
  \includegraphics{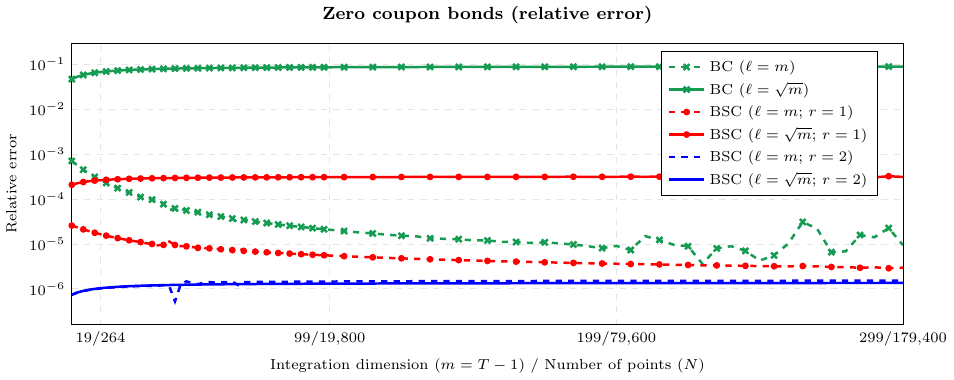}
  \caption{Numerical computation of the integral~\eqref{eq:zcbInt} using fully symmetric Bayesian cubature (BC) and Bayes--Sard cubature (BSC) for different choices of the length-scale $\ell$ and polynomial degree $r$ of the parameteric function space~$\pi$ used in BSC.
}\label{fig:zcb}
\end{figure*}

\rev{The choice of generator vectors $\Lambda = \{ \b{\lambda}^1,\ldots,\b{\lambda}^J \}$ for a fully symmetric point set is practically important and has not yet been discussed.
In principle one may wish to select $\Lambda$ in order to minimise a criterion, such as the posterior standard deviation $\sigma_N$.
However, it appears that such optimal $\Lambda$ are mathematically intractable in general.
Moreover, numerical optimisation methods cannot be naively applied to approximate the optimal $\Lambda$, since in high dimensions a sparsity structure in the generator vectors $\b{\lambda}^j$ is required to prevent creation of massive point sets $[\b{\lambda}^j]$.
Thus, although we cannot provide definitive guidelines on how to select the generators in the setting of this article, there are some useful heuristics that have guided us in the examples to follow and those presented in~\cite[Section~5]{Karvonen2018}:
\begin{itemize}
\item  In low dimensions, say $m \leq 4$, it is feasible to use (quasi) Monte Carlo samples as generators, as each fully symmetric set will contain at most $384$ points (see Table~\ref{table:sizes}). However, a large number of fully symmetric sets may be needed to ensure sufficient coverage of the space. This approach can work, as in Section~\ref{sec:global-illu-int}, but is occasionally prone to failure~\cite[Section~5.3]{Karvonen2018}.
\item In higher dimensions (or when a more robust design is desired), we recommend selecting a tried-and-tested fully symmetric point set, such as a sparse grid~\cite[Chapter 4]{Holtz2011}.
This can then be further modified if required, since fully symmetric sets can be added or removed at will. 
In very high dimensions, this can amount to using effectively low-dimensional generator vectors of the forms $(x_1,0,\ldots,0)$, $(x_2,0,\ldots,0)$, $(x_1,x_2,0,\ldots,0)$ and so on, for points $x_i$ that come from some classical one-dimensional integration rule, such as Gauss--Hermite or Clenshaw--Curtis.
\end{itemize}
These principles guided our choice of fully symmetric point sets in the sequel.
}

\subsection{Zero Coupon Bonds}\label{sec:zcb}

This example involves a model for zero coupon bonds that has been used to assess accuracy and robustness of the Bayes--Sard cubature and fully symmetric Bayesian cubature methods in~\cite{Karvonen2018,Karvonen2018c}.

\subsubsection{Integration Problem}

The integral of interest, arising from Euler--Maruyama discretisation of the Vasicek model, is
\begin{equation}\label{eq:zcbInt}
\begin{split}
P(0,T) &\coloneqq \mathbb{E}\Bigg[ \exp\bigg( -\Delta t \sum_{i=0}^{T-1} r_{t_i} \bigg)\Bigg] \\
&= \exp(-\Delta t r_{t_0}) \mathbb{E}\Bigg[ \exp\bigg( -\Delta t \sum_{i=1}^{T-1} r_{t_i} \bigg)\Bigg],
\end{split}
\end{equation}
where $r_{t_i}$ are particular Gaussian random variables and $\Delta t$ and $r_{t_0}$ are parameters of the integrand.
The dimension $m = T - 1$ of the integrand can be freely selected and the integral admits a convenient closed-form solution; see~\cite[Section~6.1]{Holtz2011} or~\cite[Section~5.5]{Karvonen2018} for a more complete description of this benchmark integral.

\subsubsection{Setting}

The accuracy of the standard Bayesian cubature and Bayes--Sard cubature methods was compared, for computing the integral~\eqref{eq:zcbInt} in a setting identical to that of~\cite[Section~5.5]{Karvonen2018}.
In particular, the same parameter values and point set (a sparse grid based on a certain Gauss--Hermite sequence with the origin removed), were used.
The kernel was the Gaussian kernel with length-scale $\ell > 0$:
\begin{equation} \label{eq: gaussian kernel}
k(\b{x},\b{x}') = \exp\bigg( \! - \frac{\norm[0]{\b{x}-\b{x}'}^2}{2\ell^2} \bigg).
\end{equation}
Accuracy of the two cubature methods was assessed for the heuristic length-scale choices $\ell = m$ and $\ell = \sqrt{m}$.
The linear space $\pi$ in the Bayes--Sard method, defined by the collection $\pi_\alpha$ of multi-indices, was taken to be $\pi_\alpha = \{\b{\alpha} : |\alpha| \leq r\}$ for either $r$ = 1 (linear) or $r = 2$ (quadratic) polynomials.
The dimension $T$ ranged between 20 and 300.
Since the number of points in a sparse grid depends on the dimension, the maximal $N$ used was $179,\!400$.
Theorems~\ref{thm:bc-fss} and~\ref{thm:BSC-FSS} facilitated the computation, respectively, of the standard Bayesian cubature and Bayes--Sard cubature method.
Note that, in the results that are presented next, even though $N$ increases, no convergence (or necessarily monotonicity of the error) is to be expected because the integration problem becomes more difficult as $T$ is increased.

\subsubsection{Results}

The results are depicted in Figure~\ref{fig:zcb}. 
We observe that Bayes--Sard method is much less sensitive to the length-scale choice compared to the standard Bayesian cubature method.
For instance, the selection $\ell = \sqrt{m}$ has Bayes--Sard outperform the standard Bayesian cubature by roughly three orders of magnitude.
It is also clear that, in this particular problem, the addition of more polynomial basis functions can significantly improve the integral estimates.

\begin{figure*}[th!]
\centering
  \includegraphics{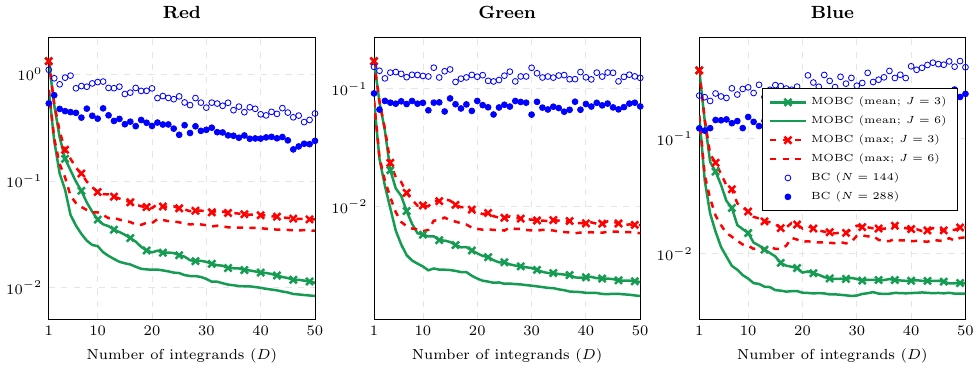}
  \caption{The mean (green; Equation~\eqref{eq:global-illu-mean}) and maximal (red; Equation~\eqref{eq:global-illu-max}) relative integration errors obtained when simultaneously approximating $D$ global illumination integrals~\eqref{eq:global-illu-int} using Bayesian cubature (BC) with random points and fully symmetric multi-output Bayesian cubature (MOBC).
Here $J=3$ and $J=6$ random generator vectors were used to produce a fully symmetric point set of size $N = 144$ (for $J = 3$) and $N = 288$ (for $J = 6$). 
The displayed results have been averaged over 100 independent realisations of the point sets.
}\label{fig:global-illumination}
\end{figure*}

Results at this scale were not possible to obtain in the earlier work of Karvonen et al.\@~\cite{Karvonen2018c}, where the largest value of $N$ considered was 5,000.
In contrast, our result in Theorem~\ref{thm:BSC-FSS} enabled point sets of size up to \sloppy{${N = 179,\!400}$} to be used.
The computational time required to produce the the results for the Bayes--Sard cubature in the most demanding case, $T = 300$ and \sloppy{${m = 2}$}, was on the order of 2.5 minutes on a standard laptop computer.
However, this can be mostly attributed to a sub-optimal algorithm for generating the sparse grid.
Indeed, after the points had been obtained it took roughly one second to compute the Bayes--Sard weights.

\subsection{Global Illumination Integrals} \label{sec:global-illu-int}

Next we considered the multi-output Bayesian cubature method, together with the symmetry exploit developed in Section~\ref{subsec: vvbc-fss}, to compute a collection of closely related integrals arising in a global illumination context.
This is a popular application of Bayesian cubature methods; see~\cite{Brouillat2009,Marques2013,Marques2015,Briol2017,Xi2018} for existing work.
In particular, multi-output Bayesian cubature was applied to the problem that we consider below in~\cite{Xi2018}, where $D = 5$ integrals were simultaneously computed.
Through computational simplifications obtained by using fully symmetric sets, in what follows we simultaneously compute up to $D = 50$ integrals, a ten-fold improvement.

\subsubsection{Integration Problem}

Global illumination is concerned with the rendering of glossy objects in a virtual environment \cite{Dutre2006}.
The integration problem studied here is to compute the \emph{outgoing radiance} $L_0(\b{\omega}_o)$ in the direction $\b{\omega}_o$, for different values of the observation angle $\b{\omega}_o$.
In practical terms, this represents the amount of light travelling from the object to an observer at an observation angle $\b{\omega}_o$.
The need for simultaneous computation for different $\b{\omega}_o$ can arise when the observation angle is rapidly changing, for example as the player moves in a video game context.
The outgoing radiance is given by the integral
\begin{equation*}
L_0(\b{\omega}_0) = L_e(\b{\omega}_0) + \int_{\mathbb{S}^2} L_i(\b{\omega}_i) \rho(\b{\omega}_i, \b{\omega}_o) [\b{\omega}_i^\T \b{n}]_+ \dif \nu(\b{\omega}_i)
\end{equation*}
with respect to the uniform (i.e.\ Riemannian) measure $\nu$ on the unit sphere 
\begin{equation*}
\mathbb{S}^2 = \big\{ \b{x} \in \mathbb{R}^3 \: \colon \norm[0]{\b{x}} = 1 \big\} \subset \mathbb{R}^3.
\end{equation*}
Here $L_e(\b{\omega}_o)$ is the amount of light emitted by the object itself, essentially a constant, whilst $L_i(\b{\omega}_i)$ is the amount of light being reflected from the object, originating from angle $\b{\omega}_i \in \mathbb{S}^2$. 
That reflection is impossible from a reflexive angle is captured by the term $[\b{\omega}_i^\T \b{n}]_+ \coloneqq \max\{0, \b{\omega}_i^\T \b{n}\}$ with $\b{n}$ the unit normal to the object.
That light is reflected less efficiently at larger incidence angles is captured by a \emph{bidirectional reflectance distribution} function
\begin{equation*}
\rho(\b{\omega}_i, \b{\omega}_o) = \frac{1}{2\pi} \exp\big( \b{\omega}_i^\T \b{\omega}_o - 1 \big) .
\end{equation*} 
Evaluation of $L_i(\b{\omega}_i)$ involves a call to an \emph{environment map} (in this case, a picture of a lake in California; see~\cite{Briol2017}), which is associated with a computational communication cost.
The illumination integral must be computed for each of the red, green, and blue (RGB) colour channels; we treat the integration problems corresponding to different colour channels as statistically independent.

\subsubsection{Setting}

The performance of the standard Bayesian cubature and multi-output Bayesian cubature methods was assessed on a collection of $D$ related integrals, where $D$ was varied up to a maximum of $D_\text{max} = 50$.
The integrands were indexed by observation angles $\b{\omega}_o^d$ with a fixed azimuth and elevation ranging uniformly on the interval $[\frac{\pi}{4} - \frac{\pi}{24}, \frac{\pi}{4} + \frac{\pi}{24}]$: 
\begin{equation*}
\b{\omega}_o^d := \Bigg(0, \frac{\pi}{4} - \frac{\pi}{24} \bigg[ 1 - 2 \left( \frac{d-1}{D_\text{max} - 1} \right) \bigg] \Bigg). 
\end{equation*}
To formulate the problem in the multi-output framework, we define the associated integrands
\begin{equation*}
f^\dagger_d = L_i(\b{\omega}_i) \rho(\b{\omega}_i, \b{\omega}_o^d) [\b{\omega}_i^\T \b{n}]_+
\end{equation*}
for $d = 1, \ldots, D_\text{max}$.
The aim is then to compute the integrals
\begin{equation}\label{eq:global-illu-int}
I(f^\dagger_d) \coloneqq \int_{\mathbb{S}^2} f^\dagger_d(\b{\omega}_i) \dif \nu(\b{\omega}_i).
\end{equation}
In our experiments a separable vector-valued covariance function was used, defined as in~\eqref{eq: sep kernel} with
\begin{equation*}
c(\b{x},\b{x}') = \frac{8}{3} - \norm[0]{\b{x}-\b{x}'}, \qquad [\b{B}]_{dq} = \exp\big( (\b{\omega}_o^d)^\T \b{\omega}_o^q - 1 \big).
\end{equation*}
This prior structure is identical to that used in~\cite{Briol2017,Xi2018} and corresponds to assuming that the integrand belongs to a Sobolev space of smoothness $\frac{3}{2}$.
The kernel $c$ has tractable kernel means: $c_\nu(\b{x}) = \frac{4}{3}$ for every $\b{x} \in \mathbb{S}^2$ and $c_{\nu,\nu} = \frac{4}{3}$.

In order to exploit Theorems~\ref{thm:bc-fss} and~\ref{thm:mobc-fss}, we need to restrict to fully symmetric point sets on $\mathbb{S}^2$.
To obtain such sets we followed the method proposed in~\cite[Section~5.3]{Karvonen2018}. 
That is, we draw, for each $d = 1, \ldots, D$, either $J = 3$ or $J = 6$ independent generator vectors from the uniform distribution $\nu$ on $\mathbb{S}^2$ and use these to generate distinct fully symmetric point sets $X_1, \ldots, X_{D_\text{max}} \subset \mathbb{S}^2$.
Equation~\eqref{eq:fss-size} implies that $N = 3 \times 48 = 144$ or $N = 6 \times 48 = 288$.\footnote{The elements of each random generator vector are almost surely non-zero and distinct.}
This approach to generation of a point set was selected for its simplicity, our main focus being on the multi-output framework and a large number of integrals $D$.
Alternative point sets on $\mathbb{S}^2$ are numerous, such as rotated adaptations of numerically computed approximations to the optimal quasi Monte Carlo designs developed in~\cite{Brauchart2014}.

\begin{figure}[t]
\centering
  \includegraphics{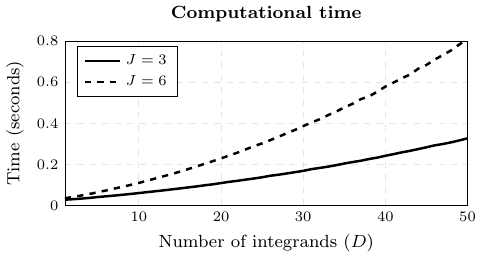}
  \caption{\rev{Average computational time (over 100 independent runs), integrand evaluations included, for computation of the fully symmetric multi-output Bayesian cubature estimates for one color channel. 
The algorithm was implemented in MATLAB and run on a desktop computer with an Intel Xeon 3.40 GHz processor and 15 GB of RAM.}
}\label{fig:global-illumination-time}
\end{figure}

\subsubsection{Results}

The results are depicted in Figure~\ref{fig:global-illumination} in terms of the relative integration error for each RGB colour channel.
For \sloppy{${D = 1,\ldots,D_\text{max}}$}, define the vector-valued functions
\begin{equation*}
\b{f}^{\dagger,D} = ( f^\dagger_1, \ldots , f^\dagger_{D} ) \colon \mathbb{S}^2 \to \mathbb{R}^{D}.
\end{equation*}
The figure shows the improvement in integration accuracy when $D$ increases and more integrands are considered simultaneously.
Displayed are the mean
\begin{equation}\label{eq:global-illu-mean}
\frac{1}{D} \sum_{d=1}^D \frac{\abs[1]{I(f^\dagger_d)-[\b{\mu}_N(\b{f}^{\dagger,D})]_d}}{I(f^\dagger_d)}
\end{equation}
and maximal
\begin{equation}\label{eq:global-illu-max}
\max_{d = 1,\ldots,D} \frac{\abs[1]{I(f^\dagger_d)-[\b{\mu}_N(\b{f}^{\dagger,D})]_d}}{I(f^\dagger_d)}
\end{equation}
relative errors for $D = 1,\ldots,D_\text{max}$.
For comparison, the figure also contains results for the standard Bayesian cubature method, applied separately to each of the uni-output integrands $f^\dagger_d$.
Each of the reference integrals $I(f^\dagger_d)$ was computed using brute force Monte Carlo, with 10 million points used.

In accordance with \cite{Xi2018}, we observed that the multi-output Bayesian cubature method is superior to the standard one already when $D = 5$.
The performance gain of the multi-output method keeps increasing when more integrands are added but is ultimately bounded.
This is reasonable since integrands for wildly different $\b{\omega}_o^d$ can convey little information about each other.
For the smallest values of $D$ the multi-output method is less accurate than the standard Bayesian cubature method.
This can be explained by potential non-uniform covering of the unit sphere when the total number $DJ$ of fully symmetric sets is low (e.g., when some of the generator vectors happen to cluster, the fully symmetric sets they generated do not greatly differ, so that less information is obtained on the integrand). \rev{For instance, the standard deviation over the 100 runs in the relative error of fully symmetric Bayesian cubature for the first integral (i.e., the case $D = 1$ in Figure~\ref{fig:global-illumination}) was 0.34 ($J=3$) or 0.17 ($J=6$) while that of the standard Bayesian cubature with random points was only 0.19 ($N=144$) or 0.11 ($N=288$). See also~\cite[Figure~5.1]{Karvonen2018}.}

Computational times remained reasonable throughout this experiment; \rev{see Figure~\ref{fig:global-illumination-time}.}
For example, without symmetry exploits, the case $D = 50$ and $J = 6$ would require $(DN)^2 = $ 207,360,000 kernel evaluations and inversion of a 14,400-dimensional matrix while Theorem~\ref{thm:mobc-fss} reduces these numbers, respectively, to \sloppy{${DNJ = 86,\!400}$} and $DJ = 300$. \rev{From Figure~\ref{fig:global-illumination-time} it is seen that this computation took only 0.8 seconds.}
This suggests that with more carefully selected fully symmetric point sets it may be possible to realise the desire expressed in~\cite[Section~4]{Xi2018} of simultaneous computation of up to thousands of related integrals.

\begin{figure}[t]
\centering
  \includegraphics{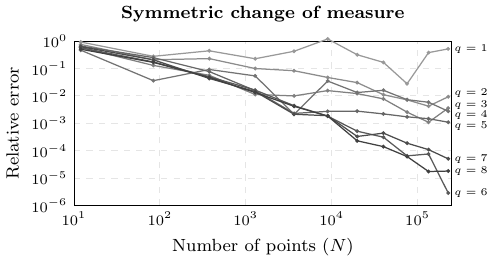}
  \caption{Relative error in numerical integration of the function~\eqref{eq: is example func} using the fully symmetric Bayesian cubature method based on a symmetric change of measure. 
\rev{Here $q$ is used to index symmetricity of $\nu$ and thus more challenging $\nu$ correspond to small $q$.}
}\label{fig:is}
\end{figure}

\subsection{Symmetric Change of Measure Illustration} \label{sec: importance sampling example}

The purpose of this final experiment is to briefly illustrate the symmetric change of measure technique, proposed in Section~\ref{sec:is-trick}.
To limit scope we consider applying this technique in conjunction with the fully symmetric standard Bayesian cubature method (i.e., Theorem~\ref{thm:bc-fss}).

\subsubsection{Integration Problem}

Let $\b{\mu}_f \in \mathbb{R}^6$ and $\b{\Sigma}_f \in \mathbb{R}^{6 \times 6}$ be a vector and a positive-definite matrix.
Consider integration over $\mathbb{R}^6$ of the function
\begin{equation} \label{eq: is example func}
f^\dagger(\b{x}) = \exp\bigg( \! -\frac{1}{2} (\b{x} - \b{\mu}_f)^\T \b{\Sigma}_f^{-1} (\b{x} - \b{\mu}_f) \bigg)
\end{equation}
with respect to a Gaussian mixture distribution, $\nu$, to be specified. 
Integrals of this form can be easily computed in closed form.
For this illustration we took
\begin{equation*}
\b{\mu}_f = \begin{bmatrix} \frac{1}{5} \\ \frac{1}{5}+\frac{3}{50} \\ \vdots \\ \frac{1}{2} - \frac{3}{50} \\ \frac{1}{2} \end{bmatrix}, \quad \b{\Sigma}_f = \begin{bmatrix} \big( \frac{4}{5} \big)^2 \b{I}_3 & \b{0}_{3 \times 3} \\ \b{0}_{3 \times 3} & \big( \frac{11}{10} \big)^2 \b{I}_3 \end{bmatrix}.
\end{equation*}

\subsubsection{Setting}

\rev{For these experiments $\nu$ was taken to be a uniform mixture of eight Gaussian distributions $\mathrm{N}(\b{\mu}_i, \b{\Sigma}_i)$, \sloppy{${i = 1,\dots,8}$}, with their mean vectors drawn independently from the standard normal distributions and detrended so that $\sum_{i=1}^8 \b{\mu}_i = \b{0}$. The covariance matrices of each Gaussian component were independent and normalised draws from the Wishart distribution $W_6(\b{I}_6, d+2(q-1))$, $q \in \mathbb{N}$. 
The resulting $\nu$ is almost surely not fully symmetric and therefore Theorem~\ref{thm:bc-fss} cannot be applied. 
Different values of $q$ correspond to different degrees of symmetricity of $\nu$: for small values of $q$ covariance matrices $\b{\Sigma}_i$ are likely to be nearly singular, while as $q \to \infty$ they become diagonal. 
Accordingly, we experimented with $q \in \{1, \ldots, 8\}$. For each $q$, the proposal distribution $\nu_*$ was a zero-mean Gaussian with diagonal covariance $\sigma^2 \b{I}_6$ for $\sigma^2$ set to the mean of the diagonal elements of the $\b{\Sigma}_i$.}
For Bayesian cubature we used the Gaussian kernel~\eqref{eq: gaussian kernel} with a length-scale $\ell = 0.8$ and the Gauss--Hermite sparse grid~\cite[Section~4.2]{Karvonen2018} with the mid-point removed.
Note that the resulting point sets are not nested for different $N$.

\subsubsection{Results}

\rev{The results of are depicted in Figure~\ref{fig:is} for one fairly representative run. Note how larger values of $q$ correspond to improved integration accuracy. It appears that for reasonably symmetric constituent distributions the proposed method works well; when the covariance matrices are nearly singular we have observed that this simple procedure can seriously fail. 
This is analogous to scenarious where standard importance sampling can be expected to fare well \cite{Robert2013}.}
Thus, based on this example at least, the symmetric change of measure technique appears to be a promising strategy to generalise the results in Theorems~\ref{thm:bc-fss},~\ref{thm:BSC-FSS} and~\ref{thm:mobc-fss}.
The largest point sets considered contained \rev{$J = 168$} fully symmetric sets, which correspond to a point set of size \rev{$N = 227,\!304$}.

\section{Discussion} \label{sec:conclusion}

There is increasing interest in the use of Bayesian methods for numerical integration \cite{Briol2017}.
Bayesian cubature methods are attractive due to analytic and theoretical tractability of the underlying Gaussian model.
However, these method are also associated with a computational cost that is cubic in the number of points, $N$, and moreover the linear systems that must be inverted are typically ill-conditioned.

The symmetry exploits developed in this work circumvent the need for large linear systems to be solved in Bayesian cubature methods.
In particular, we presented novel results for Bayes--Sard cubature \cite{Karvonen2018c} and multi-output Bayesian cubature \cite{Xi2018} that make it possible to apply these methods even for extremely large datasets or when there are many function to be integrated. 
In conjunction with the inherent robustness of the Bayes--Sard cubature method \cite{Karvonen2018c}, this results in a highly reliable probabilistic integration method that can be applied even to integrals that are relatively high-dimensional.

Three extensions of this work are highlighted:
\emph{First}, the combination of multi-output and Bayes--Sard methods appears to be a natural extension and we expect that symmetry properties can similarly be exploited for this method.
This could lead to promising procedures for integration of collections of closely related high-dimensional functions appearing in, for example, financial applications~\cite{Holtz2011}.
Similarly, our exploits should extend to the Student's $t$ based Bayesian cubatures proposed in \cite{Prueher2017}.
\emph{Second}, the investigation of optimality criteria for the symmetric change of measure technique in Section~\ref{sec:is-trick} remains to be explored.
\emph{Third}, although we focussed solely on computational aspects, the important statistical question of how to ensure Bayesian cubature methods produce output that is well-calibrated remains to some extent unresolved.\footnote{Though, see related work~\cite{Jagadeeswaran2018} on this point.}
As discussed in~\cite{Karvonen2018}, it appears that symmetry exploits do not easily lend themselves to selection of kernel parameters, for instance via cross-validation or maximisation of marginal likelihood.\footnote{An exception is for kernel amplitude parameters, which can be analytically marginalised as in Proposition~2 of~\cite{Briol2017}.}
A potential, though somewhat heuristic, way to proceed might be to exploit the concentration of measure phenomenon~\cite{Ledoux2001} or low effective dimensionality of the integrand~\cite{WangSloan2005} in order to identify a suitable data subset on which kernel parameters can be calibrated more easily or {\it a priori}.

\begin{acknowledgements}
TK was supported by the Aalto ELEC Doctoral School.
SS was supported by the Academy of Finland.
CJO was supported by the Lloyd's Register Foundation Programme on Data-Centric Engineering at the Alan Turing Institute, UK. 

This material was developed, in part, at the \textit{Prob Num 2018} workshop hosted by the Lloyd's Register Foundation programme on Data-Centric Engineering at the Alan Turing Institute, UK, and supported by the National Science Foundation, USA, under Grant DMS-1127914 to the Statistical and Applied Mathematical Sciences Institute. 
Any opinions, findings, conclusions or recommendations expressed in this material are those of the author(s) and do not necessarily reflect the views of the above-named funding bodies and research institutions.
\end{acknowledgements}


\begin{thebibliography}{10}
\providecommand{\url}[1]{{#1}}
\providecommand{\urlprefix}{URL }
\expandafter\ifx\csname urlstyle\endcsname\relax
  \providecommand{\doi}[1]{DOI~\discretionary{}{}{}#1}\else
  \providecommand{\doi}{DOI~\discretionary{}{}{}\begingroup
  \urlstyle{rm}\Url}\fi

\bibitem{Bach2017}
Bach, F.: On the equivalence between kernel quadrature rules and random feature
  expansions.
\newblock Journal of Machine Learning Research \textbf{18}(21), 1--38 (2017).
\newblock \urlprefix\url{http://www.jmlr.org/papers/volume18/15-178/15-178.pdf}

\bibitem{Bach2012}
Bach, F., Lacoste-Julien, S., Obozinski, G.: On the equivalence between herding
  and conditional gradient algorithms.
\newblock In: Proceedings of the 29th International Conference on Machine
  Learning, pp. 1355--1362 (2012).
\newblock \urlprefix\url{https://icml.cc/2012/papers/683.pdf}

\bibitem{Berlinet2011}
Berlinet, A., Thomas-Agnan, C.: Reproducing Kernel Hilbert Spaces in
  Probability and Statistics.
\newblock Springer Science \& Business Media (2011)

\bibitem{Bezhaev1991}
Bezhaev, A.{\BIBYu}.: Cubature formulae on scattered meshes.
\newblock Soviet Journal of Numerical Analysis and Mathematical Modelling
  \textbf{6}(2), 95--106 (1991).
\newblock \doi{10.1515/rnam.1991.6.2.95}

\bibitem{Brauchart2014}
Brauchart, J.S., Saff, E.B., Sloan, I.H., Womersley, R.S.: {QMC} designs:
  Optimal order quasi {M}onte {C}arlo integration schemes on the sphere.
\newblock Mathematics of Computation \textbf{83}(290), 2821--2851 (2014).
\newblock \doi{10.1090/S0025-5718-2014-02839-1}

\bibitem{Briol2017a}
Briol, F.X., Oates, C.J., Cockayne, J., Chen, W.Y., Girolami, M.: On the
  sampling problem for kernel quadrature.
\newblock In: Proceedings of the 34th International Conference on Machine
  Learning, pp. 586--595 (2017).
\newblock \urlprefix\url{http://proceedings.mlr.press/v70/briol17a.html}

\bibitem{Briol2015}
Briol, F.X., Oates, C.J., Girolami, M., Osborne, M.A.: {F}rank-{W}olfe
  {B}ayesian quadrature: Probabilistic integration with theoretical guarantees.
\newblock In: Advances in Neural Information Processing Systems, vol.~28, pp.
  1162--1170 (2015).
\newblock
  \urlprefix\url{https://papers.nips.cc/paper/5749-frank-wolfe-bayesian-quadrature-probabilistic-integration-with-theoretical-guarantees}

\bibitem{Briol2017}
Briol, F.X., Oates, C.J., Girolami, M., Osborne, M.A., Sejdinovic, D.:
  Probabilistic integration: A role in statistical computation? (with
  discussion).
\newblock Statistical Science  (2018).
\newblock \urlprefix\url{https://arxiv.org/abs/1512.00933}.
\newblock To appear.

\bibitem{Brouillat2009}
Brouillat, J., Bouville, C., Loos, B., Hansen, C., Bouatouch, K.: A {B}ayesian
  {M}onte {C}arlo approach to global illumination.
\newblock Computer Graphics Forum \textbf{28}(8), 2315--2329 (2009).
\newblock \doi{10.1111/j.1467-8659.2009.01537.x}

\bibitem{Chai2018}
Chai, H., Garnett, R.: An improved {B}ayesian framework for quadrature of
  constrained integrands  (2018).
\newblock \urlprefix\url{https://arxiv.org/abs/1802.04782}

\bibitem{Chen2018}
Chen, W., Mackey, L., Gorham, J., Briol, F.X., Oates, C.J.: Stein points.
\newblock In: Proceedings of the 35th International Conference on Machine
  Learning (2018).
\newblock \urlprefix\url{http://proceedings.mlr.press/v80/chen18f}

\bibitem{Cockayne2017}
Cockayne, J., Oates, C.J., Sullivan, T., Girolami, M.: {B}ayesian probabilistic
  numerical methods  (2017).
\newblock \urlprefix\url{https://arxiv.org/abs/1702.03673}

\bibitem{Cools1997}
Cools, R.: Constructing cubature formulae: the science behind the art.
\newblock Acta Numerica \textbf{6}, 1--54 (1997).
\newblock \doi{10.1017/S0962492900002701}

\bibitem{Davis2007}
Davis, P.J., Rabinowitz, P.: Methods of Numerical Integration.
\newblock Courier Corporation (2007)

\bibitem{DeVore2017}
DeVore, R., Foucart, S., Petrova, G., Wojtaszczyk, P.: Computing a quantity of
  interest from observational data.
\newblock Constructive Approximation  (2018).
\newblock \doi{10.1007/s00365-018-9433-7}.
\newblock To appear

\bibitem{Diaconis1988}
Diaconis, P.: Bayesian numerical analysis.
\newblock In: Statistical Decision Theory and Related Topics IV, vol.~1, pp.
  163--175. Springer-Verlag New York (1988).
\newblock \doi{10.1007/978-1-4613-8768-8_20}

\bibitem{Dietrich1997}
Dietrich, C.R., Newsam, G.N.: Fast and exact simulation of stationary gaussian
  processes through circulant embedding of the covariance matrix.
\newblock SIAM Journal on Scientific Computing \textbf{18}(4), 1088--1107
  (1997).
\newblock \doi{10.1137/s1064827592240555}

\bibitem{Dutre2006}
Dutre, P., Bekaert, P., Bala, K.: Advanced global illumination.
\newblock AK Peters/CRC Press (2006).
\newblock \doi{10.1201/9781315365473}

\bibitem{Ehler2018}
Ehler, M., Graef, M., Oates, C.J.: Optimal {M}onte {C}arlo integration on
  closed manifolds.
\newblock Statistics and Computing  (2019).
\newblock \urlprefix\url{https://arxiv.org/abs/1707.04723}.
\newblock To appear

\bibitem{Genz1986}
Genz, A.: Fully symmetric interpolatory rules for multiple integrals.
\newblock SIAM Journal on Numerical Analysis \textbf{23}(6), 1273--1283 (1986).
\newblock \doi{10.1137/0723086}

\bibitem{GenzKeister1996}
Genz, A., Keister, B.D.: Fully symmetric interpolatory rules for multiple
  integrals over infinite regions with {G}aussian weight.
\newblock Journal of Computational and Applied Mathematics \textbf{71}(2),
  299--309 (1996).
\newblock \doi{10.1016/0377-0427(95)00232-4}

\bibitem{Gunter2014}
Gunter, T., Osborne, M.A., Garnett, R., Hennig, P., Roberts, S.J.: Sampling for
  inference in probabilistic models with fast {B}ayesian quadrature.
\newblock In: Advances in Neural Information Processing Systems, vol.~27, pp.
  2789--2797 (2014).
\newblock
  \urlprefix\url{https://papers.nips.cc/paper/5483-sampling-for-inference-in-probabilistic-models-with-fast-bayesian-quadrature}

\bibitem{Hackbusch1999}
Hackbusch, W.: A sparse matrix arithmetic based on $\mathcal{H}$-matrices. part
  {I}: Introduction to $\mathcal{H}$-matrices.
\newblock Computing \textbf{62}(2), 89--108 (1999).
\newblock \doi{10.1007/s006070050015}

\bibitem{Hennig2015}
Hennig, P., Osborne, M.A., Girolami, M.: Probabilistic numerics and uncertainty
  in computations.
\newblock Proceedings of the Royal Society of London A: Mathematical, Physical
  and Engineering Sciences \textbf{471}(2179) (2015).
\newblock \doi{10.1098/rspa.2015.0142}

\bibitem{Hensman2018}
Hensman, J., Durrande, N., Solin, A.: Variational {F}ourier features for
  {G}aussian processes.
\newblock Journal of Machine Learning Research \textbf{11}(151), 1--52 (2018).
\newblock
  \urlprefix\url{https://www.jmlr.org/papers/volume18/16-579/16-579.pdf}

\bibitem{Holtz2011}
Holtz, M.: Sparse Grid Quadrature in High Dimensions with Applications in
  Finance and Insurance.
\newblock No.~77 in Lecture Notes in Computational Science and Engineering.
  Springer (2011).
\newblock \doi{10.1007/978-3-642-16004-2}

\bibitem{Jagadeeswaran2018}
Jagadeeswaran, R., Hickernell, F.J.: Fast automatic {B}ayesian cubature using
  lattice sampling  (2018).
\newblock \urlprefix\url{https://arxiv.org/abs/1809.09803}.
\newblock Submitted

\bibitem{Kanagawa2016}
Kanagawa, M., Sriperumbudur, B.K., Fukumizu, K.: Convergence guarantees for
  kernel-based quadrature rules in misspecified settings.
\newblock In: Advances in Neural Information Processing Systems, vol.~29, pp.
  3288--3296 (2016).
\newblock \urlprefix\url{https://arxiv.org/abs/1605.07254}

\bibitem{Kanagawa2017}
Kanagawa, M., Sriperumbudur, B.K., Fukumizu, K.: Convergence analysis of
  deterministic kernel-based quadrature rules in misspecified settings.
\newblock Foundations of Computational Mathematics  (2019).
\newblock \doi{10.1007/s10208-018-09407-7}.
\newblock To appear

\bibitem{Karvonen2017a}
Karvonen, T., S{\"a}rkk{\"a}, S.: Classical quadrature rules via {G}aussian
  processes.
\newblock In: 27th IEEE International Workshop on Machine Learning for Signal
  Processing (2017).
\newblock \doi{10.1109/mlsp.2017.8168195}

\bibitem{Karvonen2018}
Karvonen, T., S{\"a}rkk{\"a}, S.: Fully symmetric kernel quadrature.
\newblock SIAM Journal on Scientific Computing \textbf{40}(2), A697--A720
  (2018).
\newblock \doi{10.1137/17m1121779}

\bibitem{Karvonen2018c}
Karvonen, T., S\"{a}rkk\"{a}, S., Oates, C.J.: A {B}ayes--{S}ard cubature
  method.
\newblock In: Advances in Neural Information Processing Systems, vol.~31
  (2018).
\newblock
  \urlprefix\url{https://papers.nips.cc/paper/7829-a-bayes-sard-cubature-method}.
\newblock To appear

\bibitem{Kennedy1998}
Kennedy, M.: {B}ayesian quadrature with non-normal approximating functions.
\newblock Statistics and Computing \textbf{8}(4), 365--375 (1998).
\newblock \doi{10.1023/A:1008832824006}

\bibitem{Larkin1972}
Larkin, F.M.: Gaussian measure in {H}ilbert space and applications in numerical
  analysis.
\newblock The Rocky Mountain Journal of Mathematics \textbf{2}(3), 379--421
  (1972).
\newblock \doi{10.1216/rmj-1972-2-3-379}

\bibitem{Larkin1974}
Larkin, F.M.: Probabilistic error estimates in spline interpolation and
  quadrature.
\newblock In: Information Processing 74 (Proceedings of IFIP Congress,
  Stockholm, 1974), vol.~74, pp. 605--609. North-Holland (1974)

\bibitem{LazaroGredilla2010}
L{\'a}zaro-Gredilla, M., Qui{\~n}onero-Candela, J., Rasmussen, C.E.,
  Figueiras-Vidal, A.R.: Sparse spectrum {G}aussian process regression.
\newblock Journal of Machine Learning Research \textbf{11}, 1865--1881 (2010).
\newblock
  \urlprefix\url{http://jmlr.csail.mit.edu/papers/v11/lazaro-gredilla10a.html}

\bibitem{Ledoux2001}
Ledoux, M.: The Concentration of Measure Phenomenon.
\newblock No.~89 in Mathematical Surveys and Monographs. American Mathematical
  Society (2001).
\newblock \doi{10.1090/surv/089}

\bibitem{LuDarmofal2004}
Lu, J., Darmofal, D.L.: Higher-dimensional integration with {G}aussian weight
  for applications in probabilistic design.
\newblock SIAM Journal on Scientific Computing \textbf{26}(2), 613--624 (2004).
\newblock \doi{10.1137/s1064827503426863}

\bibitem{Marques2013}
Marques, R., Bouville, C., Ribardi{\`{e}}re, M., Santos, L.P., Bouatouch, K.: A
  spherical {G}aussian framework for {B}ayesian {M}onte {C}arlo rendering of
  glossy surfaces.
\newblock IEEE Transactions on Visualization and Computer Graphics
  \textbf{19}(10), 1619--1932 (2013).
\newblock \doi{10.1109/tvcg.2013.79}

\bibitem{Marques2015}
Marques, R., Bouville, C., Santos, L.P., Bouatouch, K.: Efficient Quadrature
  Rules for Illumination Integrals: from Quasi {M}onte {C}arlo to {B}ayesian
  {M}onte {C}arlo.
\newblock Synthesis Lectures on Computer Graphics and Animation. Morgan \&
  Claypool Publishers (2015).
\newblock \doi{10.2200/s00649ed1v01y201505cgr019}

\bibitem{McNameeStenger1967}
McNamee, J., Stenger, F.: Construction of fully symmetric numerical integration
  formulas.
\newblock Numerische Mathematik \textbf{10}(4), 327--344 (1967).
\newblock \doi{10.1007/BF02162032}

\bibitem{Minka2000}
Minka, T.: Deriving quadrature rules from {Gaussian} processes.
\newblock Tech. rep., Microsoft Research, Statistics Department, Carnegie
  Mellon University (2000).
\newblock
  \urlprefix\url{https://www.microsoft.com/en-us/research/publication/deriving-quadrature-rules-gaussian-processes/}

\bibitem{Najm2009}
Najm, H.N., Debusschere, B.J., Marzouk, Y.M., Widmer, S., Le~Ma{\^\i}tre, O.:
  Uncertainty quantification in chemical systems.
\newblock International Journal for Numerical Methods in Engineering
  \textbf{80}(6-7), 789--814 (2009).
\newblock \doi{10.1002/nme.2551}

\bibitem{NovakRitter1999}
Novak, E., Ritter, K.: Simple cubature formulas with high polynomial exactness.
\newblock Constructive Approximation \textbf{15}(4), 499--522 (1999).
\newblock \doi{10.1007/s003659900119}

\bibitem{NovakRitter1999b}
Novak, E., Ritter, K., Schmitt, R., Steinbauer, A.: On an interpolatory method
  for high dimensional integration.
\newblock Journal of Computational and Applied Mathematics \textbf{112}(1--2),
  215--228 (1999).
\newblock \doi{10.1016/s0377-0427(99)00222-8}

\bibitem{Oates2017}
Oates, C.J., Niederer, S., Lee, A., Briol, F.X., Girolami, M.: Probabilistic
  models for integration error in the assessment of functional cardiac models.
\newblock In: Advances in Neural Information Processing Systems, vol.~30, pp.
  109--117 (2017).
\newblock
  \urlprefix\url{http://papers.nips.cc/paper/6616-probabilistic-models-for-integration-error-in-the-assessment-of-functional-cardiac-models}

\bibitem{Oettershagen2017}
Oettershagen, J.: Construction of optimal cubature algorithms with applications
  to econometrics and uncertainty quantification.
\newblock Ph.D. thesis, Institut f\"{u}r Numerische Simulation, Universit\"{a}t
  Bonn (2017)

\bibitem{OHagan1978}
O'Hagan, A.: Curve fitting and optimal design for prediction.
\newblock Journal of the Royal Statistical Society. Series B (Methodological)
  \textbf{40}(1), 1--42 (1978).
\newblock \doi{10.1111/j.2517-6161.1978.tb01643.x}

\bibitem{OHagan1991}
O'Hagan, A.: Bayes--{H}ermite quadrature.
\newblock Journal of Statistical Planning and Inference \textbf{29}(3),
  245--260 (1991).
\newblock \doi{10.1016/0378-3758(91)90002-v}

\bibitem{Osborne2012}
Osborne, M., Garnett, R., Ghahramani, Z., Duvenaud, D.K., Roberts, S.J.,
  Rasmussen, C.E.: Active learning of model evidence using {B}ayesian
  quadrature.
\newblock In: Advances in Neural Information Processing Systems, vol.~25, pp.
  46--54 (2012).
\newblock
  \urlprefix\url{https://papers.nips.cc/paper/4657-active-learning-of-model-evidence-using-bayesian-quadrature}

\bibitem{Osborne2012a}
Osborne, M., Garnett, R., Roberts, S., Hart, C., Aigrain, S., Gibson, N.:
  Bayesian quadrature for ratios.
\newblock In: Artificial Intelligence and Statistics, pp. 832--840 (2012).
\newblock
  \urlprefix\url{http://proceedings.mlr.press/v22/osborne12/osborne12.pdf}

\bibitem{Pronzato2018}
Pronzato, L., Zhigljavsky, A.: Bayesian quadrature and energy minimization for
  space-filling design  (2018).
\newblock \urlprefix\url{https://arxiv.org/abs/1808.10722}

\bibitem{Prueher2017}
Pr{\"u}her, J., Tronarp, F., Karvonen, T., S{\"a}rkk{\"a}, S., Straka, O.:
  {S}tudent-$t$ process quadratures for filtering of non-linear systems with
  heavy-tailed noise.
\newblock In: 20th International Conference on Information Fusion (2017).
\newblock \doi{10.23919/icif.2017.8009742}

\bibitem{Rasmussen2006}
Rasmussen, C.E., Williams, C.K.I.: Gaussian Processes for Machine Learning.
\newblock MIT Press (2006)

\bibitem{Robert2013}
Robert, C., Casella, G.: Monte Carlo Statistical Methods.
\newblock Springer Science \& Business Media (2013)

\bibitem{Sarkka2016}
S{\"a}rkk{\"a}, S., Hartikainen, J., Svensson, L., Sandblom, F.: On the
  relation between {G}aussian process quadratures and sigma-point methods.
\newblock Journal of Advances in Information Fusion \textbf{11}(1), 31--46
  (2016).
\newblock \urlprefix\url{https://arxiv.org/abs/1504.05994}

\bibitem{Schaback1995}
Schaback, R.: Error estimates and condition numbers for radial basis function
  interpolation.
\newblock Advances in Computational Mathematics \textbf{3}(3), 251--264 (1995).
\newblock \doi{10.1007/bf02432002}

\bibitem{Schaefer2017}
Sch{\"a}fer, F., Sullivan, T.J., Owhadi, H.: Compression, inversion, and
  approximate {PCA} of dense kernel matrices at near-linear computational
  complexity  (2017).
\newblock \urlprefix\url{https://arxiv.org/abs/1706.02205}

\bibitem{Smola2007}
Smola, A., Gretton, A., Song, L., Sch{\"o}lkopf, B.: A {H}ilbert space
  embedding for distributions.
\newblock In: International Conference on Algorithmic Learning Theory, pp.
  13--31. Springer (2007).
\newblock \doi{10.1007/978-3-540-75225-7_5}

\bibitem{SommarivaVianello2006}
Sommariva, A., Vianello, M.: Numerical cubature on scattered data by radial
  basis functions.
\newblock Computing \textbf{76}(3--4), 295--310 (2006).
\newblock \doi{10.1007/s00607-005-0142-2}

\bibitem{Stein2012}
Stein, M.L.: Interpolation of Spatial Data: Some Theory for Kriging.
\newblock Springer Science \& Business Media (2012)

\bibitem{WangSloan2005}
Wang, X., Sloan, I.H.: Why are high-dimensional finance problems often of low
  effective dimension?
\newblock SIAM Journal on Scientific Computing \textbf{27}(1), 159--183 (2005).
\newblock \doi{10.1137/s1064827503429429}

\bibitem{Wendland2005}
Wendland, H.: Scattered Data Approximation.
\newblock No.~28 in Cambridge Monographs on Applied and Computational
  Mathematics. Cambridge University Press (2005)

\bibitem{Xi2018}
Xi, X., Briol, F.X., Girolami, M.: {B}ayesian quadrature for multiple related
  integrals.
\newblock Proceedings of the 35th International Conference on Machine Learning
  (2018).
\newblock \urlprefix\url{https://arxiv.org/abs/1801.04153}.
\newblock To appear

\bibitem{Xiu2003}
Xiu, D., Karniadakis, G.E.: Modeling uncertainty in flow simulations via
  generalized polynomial chaos.
\newblock Journal of Computational Physics \textbf{187}(1), 137--167 (2003).
\newblock \doi{10.1016/s0021-9991(03)00092-5}

\bibitem{Alvarez2012}
Álvarez, M., Rosasco, L., Lawrence, N.: Kernels for vector-valued functions: A
  review.
\newblock Foundations and Trends\textregistered in Machine Learning
  \textbf{4}(3), 195--266 (2012).
\newblock \doi{10.1561/2200000036}

\end{thebibliography}

\providecommand{\BIBYu}{Yu}

\end{document}